\documentclass[11pt,a4paper]{article}

\usepackage[T1]{fontenc}
\usepackage[utf8]{inputenc}
\usepackage[english]{babel}
\usepackage[margin=1in]{geometry}
\usepackage{authblk}
\usepackage{graphicx}
\usepackage{amsmath,amssymb,amsfonts}
\usepackage{mathtools}
\usepackage{bm}
\usepackage{booktabs}
\usepackage{float}
\usepackage{capt-of}
\usepackage{placeins}
\usepackage{subcaption}
\usepackage{algorithm}
\usepackage{algpseudocode}
\usepackage{amsthm}
\usepackage{microtype}
\usepackage{xspace}
\usepackage{etoolbox}
\usepackage{xcolor}
\usepackage[authoryear,longnamesfirst]{natbib}
\usepackage[hidelinks]{hyperref}

\def\tsc#1{\csdef{#1}{\textsc{\lowercase{#1}}\xspace}}
\tsc{WGM}
\tsc{QE}

\newtheorem{theorem}{Theorem}

\newtheorem{lemma}[theorem]{Lemma}

\theoremstyle{definition}
\newtheorem{definition}{Definition}
\newcommand{\RR}{\mathbb{R}}
\newcommand{\EE}{\mathbb{E}}
\newcommand{\1}{\mathbb{I}}

\DeclareMathOperator{\rhoSpec}{\rho}

\definecolor{colNOS}{RGB}{31,119,180}
\definecolor{coltGNN}{RGB}{255,127,14}
\definecolor{colGRU}{RGB}{44,160,44}
\definecolor{colRNN}{RGB}{214,39,40}
\definecolor{colMLP}{RGB}{148,103,189}

\makeatletter
\g@addto@macro\normalsize{%
  \setlength{\abovedisplayskip}{1pt}%
  \setlength{\belowdisplayskip}{1pt}%
  \setlength{\abovedisplayshortskip}{0pt}%
  \setlength{\belowdisplayshortskip}{0pt}%
}
\makeatother

\title{Certified Closed-Loop Control for Packet Networks: A Compositional Certification Framework}

\author[1]{Muhammad Bilal\thanks{Corresponding author: Muhammad Bilal. Email: \texttt{m.bilal@ieee.org}}}
\author[2]{Jon Crowcroft}
\author[3]{Xiaolong Xu}
\author[4]{Huaming Wu}

\affil[1]{School of Computing and Communications, Lancaster University, United Kingdom}
\affil[2]{Department of Computer Science and Technology, University of Cambridge, Cambridge, United Kingdom\\\texttt{jon.crowcroft@cl.cam.ac.uk}}
\affil[3]{School of Computer Science, School of Software, Nanjing University of Information Science and Technology, China\\\texttt{xlxu@ieee.org}}
\affil[4]{Center for Applied Mathematics, Tianjin University, China\\\texttt{whming@tju.edu.cn}}

\date{}

\begin{document}

\maketitle

\begin{abstract}
Packet networks are controlled dynamical systems with discontinuities, delayed observations, and partial state information. Adaptive or learning-driven proposers can improve performance, but an unsafe proposal may still cause starvation, tail-delay spikes, or unstable queue behaviour. This paper treats packet-network control as an executed-action certification problem. A certified operator sits between any proposer and the dataplane. At each control tick, the proposer emits an arbitrary candidate action $\tilde u(t)$. The operator either projects it to an executable action $u(t)$ that satisfies a configuration-compiled certificate, or reports \texttt{INFEASIBLE} and executes an always-defined fallback with quantified slack. The certificate also exports an auditable envelope $\bar z(t)$ for downstream composition.
The guarantees are conditional and explicit. They apply on ticks where the operator reports \texttt{CERTIFIED}, the declared arrival envelope and backlog bound are valid, and the platform realises the assumed service lower bound. Under these conditions, one mechanism covers backlog caps, service floors, mitigation caps, Foster--Lyapunov drift constraints, and compositional envelope contracts. We prove operator-level safety, feed-forward compositional safety and stability using exported envelopes, and a cyclic closure result under a small-gain condition. We also define breach and infeasibility semantics, discuss calibration of the service-tracking factor that links certified targets to realised scheduler behaviour, and evaluate the design under delayed telemetry, delayed actuation, weak proposers, envelope mismatch, overload, and millisecond-scale certification. The present evaluation validates the certified execution boundary in a byte-level closed-loop backend; deployment-level scheduler tracking is left to future Linux or hardware experiments.
\end{abstract}

\paragraph{Highlights}
\begin{itemize}
\item Frames packet-network control as certification of the action actually executed at the dataplane.
\item Separates arbitrary proposers from the certified operator, status flag, slack report, and exported envelope contract.
\item Provides per-tick constraints for backlog caps, service floors, mitigation caps, and drift-based queue stability.
\item Gives compositional safety and stability results for feed-forward networks, plus cyclic closure diagnostics under a small-gain condition.
\item Makes the limits of certification explicit through valid-envelope, valid-state-bound, service-tracking, breach, and infeasibility semantics.
\end{itemize}

\paragraph{Keywords}
Certified closed-loop network control; Packet networks; Safety filtering; Runtime assurance; Compositional certification; Queue stability

\section{Introduction}
% =========================

Packet networks are controlled dynamical systems in which scheduling and queueing decide who is served and when, queue management decides what is marked or dropped, routing determines where work flows, and mitigation actions shape or isolate traffic. These mechanisms connect continuous objectives, such as tail delay, fairness, and delivered rate, to discontinuous packet-level dynamics, including integer service, idling, and reflection at zero backlog. Classical mechanisms already show that small control choices can have large effects on delay and stability, including fair queueing, deficit round robin, active queue management, and service abstractions for per-flow guarantees \cite{demers1989fairqueue,shreedhar1995drr,floyd1993red,parekh1993gps1,parekh1994gps2}.

Queue-drift analysis has long been used to reason about packet-network stability \cite{tassiulas1992stability,neely2010stochastic}. In this view, the relevant object is the service action executed by the system, rather than the controller’s intended action. That distinction matters in modern packet systems. Data-centre and edge networks often operate at high utilisation, with tail-latency targets and feedback drawn from sampled counters, aggregate telemetry, or in-band summaries rather than direct queue state \cite{alizadeh2010dctcp,perry2014fastpass,li2019hpcc,benbasat2020pint,shriv2022,jeyakumar2024umon}. Programmable monitoring improves visibility, but it still leaves the controller making some decisions before the full state is known.

Network control is also becoming more adaptive. Recent work includes synthesis-style control, learning from expert behaviour, and controllers trained across varied network conditions \cite{winstein2013tcpex,emara2020eagle,zhou2022deepcc,xia2022genet,bilal2025nosoran}. The risk is not only poor average performance. Model error, traffic shift, or timed adversarial behaviour can produce starvation, delay spikes, or unstable control actions. Delayed telemetry and delayed actuation make this worse because the controller may recognise the failure only after queues have already grown. Even under simple admissibility rules, worst-case arrivals can increase delay and backlog in ways that average-case reasoning does not predict \cite{borodin2001aqt}.

This paper asks a deployment-facing question: how can adaptive proposing be combined with strict, online-checkable guarantees for packet queues, in a form that composes across a network?
% We separate proposing from correctness by placing a \emph{certified operator} between any proposer and the dataplane. At each control tick, the proposer emits a candidate action $\tilde u(t)$ based on delayed telemetry; in agentic NetOps settings, this proposer may also be a software agent or LLM-based assistant \cite{bilal2026netOps}. The operator mediates execution by enforcing constraints compiled from a configuration object. It returns an executed action $u(t)$, an exported envelope $\bar z(t)$ for downstream composition, and a status flag recording whether certification succeeded.
We separate proposal from correctness by placing a \emph{certified operator} between the proposer and the dataplane. At each control tick, the proposer emits a candidate action \(\tilde u(t)\) from delayed telemetry. In agentic NetOps settings, it may also be a software agent or LLM-based assistant \cite{bilal2026netOps}. The operator then enforces constraints compiled from configuration and returns the executed action \(u(t)\), an exported envelope \(\bar z(t)\) for downstream composition, and a status flag indicating whether certification succeeded.

The guarantees attach to the executed action $u(t)$, not to the internal logic of the proposer.
This ``execute through a filter'' pattern is close to safety filtering in control, where an arbitrary controller is wrapped by an online feasibility check and a correction step \cite{ames2017cbfqp,wabersich2021psf}.
For composition, we use an envelope-contract view.
Modules exchange certified bounds, rather than internal policies.
This is adjacent to deterministic envelope reasoning in network calculus, with the difference that the bounds here are enforced at runtime through the operator interface \cite{cruz1991calculus1,cruz1991calculus2}.

The paper makes four contributions.
First, it formalises a delayed closed-loop model for packet queues and defines a certified-operator interface that compiles a unified, configuration-driven certificate into per-tick constraints.
Second, it supports backlog-cap and floor/cap safety constraints, drift-based stability guarantees with explicit constants, and compositional contracts through exported envelopes.
Third, it proves operator-level safety, compositional safety, and stability for feed-forward DAG networks using only exported envelopes.
It also gives a cyclic-network closure result based on an envelope fixed point under a small-gain condition.
Fourth, it specifies contract-breach semantics and an always-defined emergency fallback for infeasible constraints.
The evaluation covers timing-shaped stress, envelope mismatch, overload, delayed telemetry and actuation, and pinned control-tick runtime microbenchmarks \cite{bilal2026gate}.

The envelope layer is designed to work with classical network calculus, not to replace it. Declared arrival envelopes $\bar A(t)$ and exported envelopes $\bar z(t)$ may be instantiated from arrival-curve and service-curve assumptions, with token-bucket contracts as a standard case \cite{leboudec2001network,cruz1991calculus1}. Classical network calculus uses such assumptions to derive deterministic backlog and delay bounds for queueing elements and their composition. Here the emphasis is different: envelopes are used inside an online certified operator that filters arbitrary proposed actions at each control tick. The paper proves operator-level and compositional properties for the abstract queueing model and for the compiled per-tick constraints. It does not claim mechanised verification of the certificate compiler, solver, runtime, or dataplane backend. Strengthening that trusted computing base through translation validation or machine-checked implementation proofs is a natural next step, but outside the scope of this paper \cite{sewell2013translation,klein2014comprehensive}.

\paragraph{Positioning and scope.}

Queueing stability and drift-based control provide the mathematical backbone for many packet-network controllers. Max-weight style arguments and Foster--Lyapunov drift give conditions under which queues remain stable under admissible loads \cite{tassiulas1992stability,neely2010stochastic}. The present paper uses the same discipline, but moves the proof obligation to the action that is executed after online certification. The proposer is not trusted to preserve drift by itself.

Network calculus and service-curve methods give deterministic backlog and delay bounds from arrival and service envelopes \cite{cruz1991calculus1,cruz1991calculus2,leboudec2001network}. Our envelope contracts are deliberately compatible with that tradition. The difference is operational: the envelope is not only an offline assumption used to derive a static bound. It is also an input to a per-tick certification interface that accepts, corrects, or rejects a proposed dataplane action.

Safety filters, shielding, and control-barrier-function methods similarly wrap an arbitrary controller with online feasibility checks \cite{ames2017cbfqp,ames2019cbf,wabersich2021psf}. Packet queues make this setting less tidy than smooth continuous-state systems. Service is packetised, queues reflect at zero, actuation may be delayed, and the scheduler may realise only a lower-bounded fraction of the requested service target. The certified operator therefore exposes these assumptions as part of its status and diagnostic interface.

Verified systems and dataplane-correctness work address a different layer of trust. Translation validation and machine-checked implementations can reduce the trusted computing base of the compiler or runtime \cite{sewell2013translation,klein2014comprehensive}. This paper does not claim such mechanised verification. It uses these traditions to define a runtime certification boundary for executed packet-network actions. The certificate compiler, solver implementation, and dataplane backend remain trusted components unless separately verified.

\begin{center}
\begin{minipage}{0.96\linewidth}
\small
\textbf{Core distinction.}
Existing adaptive control work often analyses a policy, controller, or intended action. This paper analyses only the action actually executed after certification. The proposer may be heuristic, learned, adversarial, or agentic. Correctness attaches to the certified output \(u(t)\), the status flag \(\sigma(t)\), the diagnostic report \(r(t)\), and the exported envelope contract \(\bar z(t)\).
\end{minipage}
\end{center}

\begin{table}[t]
\centering
\small
\caption{Guarantee semantics exposed by the certified operator.}
\label{tab:guarantee_semantics}
\begin{tabular}{p{0.19\linewidth}p{0.25\linewidth}p{0.34\linewidth}p{0.15\linewidth}}
\toprule
Status or flag & Action taken & Guarantee available & Logged evidence \\
\midrule
\texttt{CERTIFIED} & Projected action \(u(t)\) & Safety and drift guarantees, provided the declared envelope, backlog bound, and service-tracking lower bound are valid & Envelope, action, status \\
\texttt{INFEASIBLE} & Slack-minimising fallback & No hard guarantee; quantified violation under the compiled constraints & Slack vector \\
\texttt{MISSING\_ENVELOPE} & Hold-safe or fallback & No compositional guarantee for that tick & Missing contract \\
\texttt{MISSING\_STATEBOUND} & Hold-safe or fallback & No state-based safety guarantee for that tick & Missing state bound \\
Breach flag \(b(t)=1\) & Action may still run & Downstream composition is not valid for that tick & Breach flag \\
\bottomrule
\end{tabular}
\end{table}

% =========================
\section{Model and platform realisation}
\label{sec:model}
% =========================

\paragraph{Queues and actions.}
Time is slotted $t=0,1,2,\dots$ with step $\Delta$. For a module (node or link scheduler) $M$, let $q_i^M(t)\in \RR_{\ge 0}$ denote backlog of class-queue $i\in\{1,\dots,N_M\}$ at time $t$.

Arrivals to module $M$ consist of exogenous arrivals $a_i^M(t)$ and routed inflow $\mathrm{in}_i^M(t)$. The module executes service $\mu_i^M(t)\ge 0$ and shedding $\delta_i^M(t)\ge 0$ (drop, admission denial, or policing). Define total removal
\begin{equation}
s_i^M(t) \triangleq \mu_i^M(t)+\delta_i^M(t).
\label{eq:s_def}
\end{equation}
The queue update is
\begin{equation}
q_i^M(t+1) \;=\; \big[q_i^M(t) + a_i^M(t) + \mathrm{in}_i^M(t) - s_i^M(t)\big]_+.
\label{eq:q_update_module}
\end{equation}

Resource constraints define an admissible set. For a bottleneck resource at module $M$ with capacity $C_M(t)$,
\begin{equation}
\sum_{i=1}^{N_M} \mu_i^M(t) \le C_M(t),
\qquad \mu_i^M(t)\ge 0,\qquad \delta_i^M(t)\ge 0.
\label{eq:cap_module}
\end{equation}
Additional constraints (WFQ weight bounds, class caps, service floors) are modelled as linear inequalities in $\mu^M(t)$.

\paragraph{Delayed telemetry and delayed actuation.}
The proposer observes telemetry $y^M(t)$ rather than $q^M(t)$. A delayed, noisy measurement model is
\begin{equation}
y^M(t) = H_M q^M(t-\tau_y^M) + \nu^M(t),
\label{eq:telemetry}
\end{equation}
where $H_M$ aggregates counters, $\tau_y^M\ge 0$ is telemetry delay, and $\nu^M(t)$ is measurement noise.
Section~\ref{sec:delay} describes robust handling.

\paragraph{Arrival envelopes and adversaries.}
We assume declared envelopes on arrivals:
\begin{equation}
0 \le a_i^M(t) \le \bar a_i^M(t),\qquad
0 \le \mathrm{in}_i^M(t) \le \bar{\mathrm{in}}_i^M(t).
\label{eq:envelopes}
\end{equation}
Define the total envelope
\begin{equation}
\bar A_i^M(t) \triangleq \bar a_i^M(t) + \bar{\mathrm{in}}_i^M(t).
\label{eq:total_env}
\end{equation}
Adversarial timing shaping is captured by allowing arrivals to be any sequence respecting \eqref{eq:envelopes}. Correctness is with respect to the declared envelope.
This envelope-conditioned viewpoint matches a long line of drift-based queueing control, and it makes the distinction between executed actions and proposed actions explicit in the object that is analysed \cite{neely2010stochastic,tassiulas1992stability}. The use of envelope-bounded but otherwise unconstrained timing also aligns with worst-case perspectives studied in adversarial queueing \cite{borodin2001aqt}.

\subsection{Service targets and realised removals on commodity schedulers}
\label{app:tracking}

The model in \eqref{eq:q_update_module} uses $\mu_i(t)$ as the realised service removal on tick $t$.
On real platforms, the controller often sets weights or rates (WFQ, DRR, HTB), and realised per-tick removals can differ from targets because of packetisation, quantum effects, non-backlogged classes, variable link rate, and cross-traffic.

To make this explicit, we interpret the operator output as a \emph{service target} $\mu_i^{\mathrm{tar}}(t)$ and allow a bounded tracking model for realised service:
\begin{equation}
\mu_i(t) \ge \kappa_i(t)\,\mu_i^{\mathrm{tar}}(t),\qquad \kappa_i(t)\in[\kappa_{\min},1],
\label{eq:service_tracking}
\end{equation}
where $\kappa_{\min}$ is a conservative platform constant (or an online lower bound) that accounts for scheduler noise and rate variability over one control tick.

The operator can stay robust by finding its lower limit on the target variable through $\kappa_{\min}$.

Concretely, if certification requires $s_i(t)=\mu_i(t)+\delta_i(t)\ge \ell_i(t)$, it is sufficient to enforce
\begin{equation}
\kappa_{\min}\mu_i^{\mathrm{tar}}(t) + \delta_i(t) \ge \ell_i(t).
\label{eq:tracking_robust_lb}
\end{equation}
While this approach is cautious, it reflects how deployments actually operate. We authenticate based on the functionality that the platform can guarantee, rather than on an idealized clock-level scheduler.
In the experiments we log the realised tracking ratio $\mu_i(t)/\mu_i^{\mathrm{tar}}(t)$ to show when this conservatism matters.

\paragraph{Estimating \(\kappa_{\min}\) in deployment.}
The service-tracking factor should be treated as a calibrated platform parameter, not as a free modelling convenience. During a calibration window of \(W_\kappa\) ticks, the controller logs the target service \(\mu_i^{\mathrm{tar}}(t)\) and the realised removal \(\mu_i(t)\) for each backlogged class. It then computes
\begin{equation}
\widehat\kappa_i
=
Q_{p_\kappa}\!\left(
\left\{\frac{\mu_i(t)}{\max\{\mu_i^{\mathrm{tar}}(t),\eta_\mu\}}:
t\in\mathcal{W}_\kappa,\ q_i(t)>0
\right\}
\right),
\label{eq:kappa_quantile}
\end{equation}
where \(Q_{p_\kappa}\) is a low empirical quantile, for example \(p_\kappa=0.01\), and \(\eta_\mu>0\) avoids division by zero. A conservative deployment uses
\begin{equation}
\kappa_{\min}=\max\{\kappa_{\mathrm{floor}},\ \min_i \widehat\kappa_i-m_\kappa\},
\label{eq:kappa_min_calibration}
\end{equation}
with margin \(m_\kappa\) and a configured lower floor \(\kappa_{\mathrm{floor}}\). The estimate may be refreshed only when the platform is in a stable regime and the observed breach rate is low. Otherwise it is frozen. This prevents a transient overload period from silently weakening the certificate.

\begin{table}[t]
\centering
\small
\caption{Operational interpretation of the service-tracking factor. The categories are qualitative design guidance only; the present experiments do not estimate deployment-level \(\kappa_{\min}\).}
\label{tab:kappa_examples}
\begin{tabular}{p{0.31\linewidth}p{0.23\linewidth}p{0.36\linewidth}}
\toprule
Scheduler regime & Tracking reliability & Main source of loss from target \\
\midrule
Ideal fluid or long-tick scheduler & high & negligible packetisation error \\
Packetised WFQ/DRR with active classes & high to moderate & quantum, packet size, and class backlogging effects \\
HTB or paced queues under cross-traffic & moderate & rate policing, burst shaping, batching, and link variability \\
Congested software path with NIC batching & low to moderate & batching delay and short-window service jitter \\
\bottomrule
\end{tabular}
\end{table}

% -------------------------

\section{Unified certificate mechanism and configuration}
\label{sec:certs}
% =========================

At each time $t$, certification induces a certified feasible set $\mathcal{U}_{\mathrm{cert}}^M(t)$ over actions.
We express certificates in a form that can be compiled from configuration and checked online.

\paragraph{Type B: barrier certificates for safety}
Fix backlog caps $Q_{i}^{M,\max}>0$ and define the safe set
\begin{equation}
\mathcal{S}^M=\{q^M\in \RR_{\ge 0}^{N_M}: q_i^M\le Q_{i}^{M,\max}\ \forall i\}.
\label{eq:safe_set}
\end{equation}
A sufficient one-step barrier condition for forward invariance under \eqref{eq:q_update_module} and \eqref{eq:envelopes} is
\begin{equation}
q_i^M(t) + \bar A_i^M(t) - s_i^M(t) \le Q_{i}^{M,\max}\quad \forall i,
\label{eq:barrier_safety}
\end{equation}
equivalently,
\begin{equation}
s_i^M(t) \ge q_i^M(t) + \bar A_i^M(t) - Q_{i}^{M,\max}\quad \forall i.
\label{eq:barrier_safety_lb}
\end{equation}

Safety often includes linear constraints such as service floors for protected classes $\mathcal{P}^M$,
\begin{equation}
\mu_i^M(t) \ge \phi_i^M C_M(t),\quad \forall i\in\mathcal{P}^M,
\label{eq:floor}
\end{equation}
and mitigation caps for flagged classes $\mathcal{F}^M(t)$,
\begin{equation}
\mu_i^M(t) \le \rho_i^M C_M(t),\quad \forall i\in\mathcal{F}^M(t).
\label{eq:cap_mitig}
\end{equation}
Figure~\ref{fig:priority_inversion} illustrates why explicit floors/caps are needed: without them, a proposer can induce a priority inversion in which the latency class tail delay degrades sharply while the bulk class appears artificially good.

\paragraph{Type A: drift certificates for stability}
Fix high thresholds $Q_{i}^{M,\mathrm{hi}}\in(0,Q_{i}^{M,\max})$ and margins $\varepsilon_i^M>0$. The drift trigger is
\begin{equation}
q_i^M(t)\ge Q_{i}^{M,\mathrm{hi}}
\;\Rightarrow\;
s_i^M(t) \ge \bar A_i^M(t)+\varepsilon_i^M.
\label{eq:drift_trigger}
\end{equation}
This is enforced online and yields an explicit Foster--Lyapunov drift bound (Section~\ref{sec:stability}).

\paragraph{Type C: contract certificates for composition}
For each directed interconnection $U\to V$, define the realised outflow signal $z_{U\to V}(t)$ that enters downstream as inflow. A contract certificate exports an envelope
\begin{equation}
0 \le z_{U\to V}(t) \le \bar z_{U\to V}(t),
\label{eq:export_env}
\end{equation}
and downstream assumes
\begin{equation}
\bar{\mathrm{in}}^V(t)\ \ge\ \sum_{U\in\mathrm{Pred}(V)} \bar z_{U\to V}(t).
\label{eq:compose_in}
\end{equation}
Modules exchange bounds, not internal policies.

\paragraph{Contract mechanics.}
This section gives only the interface. Section~\ref{app:contracts} defines how realised inter-module flow \(z\) and exported envelopes \(\bar z\) are computed from executed removals, including the physical cap that prevents a module from exporting more work than backlog plus envelope arrivals.

% =========================
\subsection{Certificates as configuration and feasibility diagnostics}
\label{app:config}
% -------------------------

For each class $i$ at module $M$, the configuration specifies queue limits, tolerance margin, optional caps, and the action form:
\begin{align}
\begin{split}
   \Theta_i^M = \big(Q_{i}^{M,\max},\ Q_{i}^{M,\mathrm{hi}},\ \varepsilon_i^M,   \text{optional caps},\ \text{action parameterisation}\big).
\label{eq:theta_fields} 
\end{split}
\end{align}
The action parameterisation maps an executed decision to platform knobs such as WFQ weights, HTB rates, or policing rates.

\paragraph{Compilation to constraints.}
At each tick, the operator uses $\Theta^M$ and the local information set $\mathcal{I}^M(t)$ to construct the certified feasible set $\mathcal{U}_{\mathrm{cert}}^M(t;\Theta^M)$. The construction compiles these fields into:
\begin{itemize}
\item barrier constraints enforcing \eqref{eq:barrier_safety} and any optional floors or caps,
\item drift constraints enforcing \eqref{eq:drift_trigger} when active,
\item resource constraints enforcing \eqref{eq:cap_module} and action limits,
\item contract rules defining the exported envelopes in \eqref{eq:export_env}.
\end{itemize}

\textbf{Proof obligations checked by the operator.}
Theorems in this paper apply on ticks where $\sigma^M(t)=\texttt{CERTIFIED}$.
Certification requires: (i) envelopes $\bar A^M(t)$ are present and nonnegative, (ii) capacity and action limits are present,
(iii) drift activation and barrier constraints are built using $\overline q^M(t)$ when telemetry is delayed or bounded,
(iv) exported envelopes $\bar z^M(t)$ are computed from the executed action $u^M(t)$, not the proposal $\tilde u^M(t)$,
and (v) infeasibility is reported explicitly (no silent fallback).

\paragraph{A single feasibility lemma for \texttt{CERTIFIED}.}
\label{sec:feasibility_lemma}

We state one reusable feasibility condition that explains when \texttt{CERTIFIED} is expected and when \texttt{INFEASIBLE} is unavoidable.

Let $\ell_i(t)$ be the compiled per-class lower bound on total removal $s_i=\mu_i+\delta_i$ at time $t$ induced by safety and drift constraints (with the operator’s chosen state bound and envelope). Let optional per-class service floors/caps be $\mu_i\ge \underline\mu_i$ and $\mu_i\le \overline\mu_i$, and let $\delta_i\le \overline\delta_i$ be an optional shedding cap (take $\overline\delta_i=\infty$ if uncapped).

\begin{lemma}[Sufficient feasibility condition]
\label{lem:feasible_certified}
Fix $t$ and suppose the operator has the required inputs (state bound, envelope, and capacity). Define a conservative service requirement
\begin{equation}
r_i(t)\triangleq \max\{\,\underline\mu_i,\ \ell_i(t)-\overline\delta_i\,\},
\label{eq:residual_req}
\end{equation}
with the convention $\ell_i-\overline\delta_i=-\infty$ when $\overline\delta_i=\infty$.
If there exists $\mu\in\RR_{\ge 0}^N$ such that
\begin{equation}
\sum_i \mu_i \le C(t),
\qquad
r_i(t)\le \mu_i \le \overline\mu_i \ \ \forall i,
\label{eq:feasible_mu}
\end{equation}
then there exists an action $(\mu,\delta)$ satisfying all compiled linear constraints with zero slack, and hence the certified feasible set $\mathcal{U}_{\mathrm{cert}}(t)$ is nonempty.
In this case, the projection-based operator can report \texttt{CERTIFIED} (up to solver and numerical tolerances).
\end{lemma}

The condition says: after you account for any allowed shedding $\overline\delta_i$, the remaining per-class requirement must fit inside capacity and service floors/caps. With actuation delay, the apply-time compilation in Section~\ref{sec:delay_applied} increases the effective lower bounds $\ell_i(t)$ because the operator must protect against arrivals during the delay window while actions are still in flight. Thus Lemma~\ref{lem:feasible_certified} is also the diagnostic test for when certification is structurally possible at a tick.

% -------------------------

\section{Certified operator and closed-loop realisation}
\label{sec:cert_operator}
% =========================

We now define the certified operator, which is the only component allowed to
emit the executed action.

\begin{definition}[Certified operator]
\label{def:certified_operator}
Fix a certificate configuration $\Theta$ (Section~\ref{app:config}).
The certified operator is a map
\begin{equation}
\big(u^M(t), \bar z^M(t), \sigma^M(t), r^M(t)\big)
\;=\;
\mathcal{C}_{\Theta}^M\!\big(\mathcal{I}^M(t),\, \tilde u^M(t)\big),
\label{eq:cert_op}
\end{equation}
where $\mathcal{I}^M(t)$ contains a backlog estimate or bound, a total envelope
$\bar A^M(t)$, capacity and action limits, and the proposal $\tilde u^M(t)$.
The outputs are the executed action $u^M(t)$, exported envelope $\bar z^M(t)$,
status flag $\sigma^M(t)$, and diagnostic report $r^M(t)$.
\end{definition}

\paragraph{Status semantics.}
The status flag takes values in a small set:
\begin{align}
    \begin{split}
        \sigma^M(t)\in\{\texttt{CERTIFIED},\ \texttt{INFEASIBLE},
\texttt{MISSING\_ENVELOPE},\texttt{MISSING\_STATEBOUND}\}.
    \end{split}
\end{align}
Breach is recorded separately through \(b^M(t)\), because a breach concerns the validity of the declared contract rather than only the feasibility of the projection. Table~\ref{tab:guarantee_semantics} summarises the resulting guarantee semantics. All theorems in this paper apply to time steps for which \(\sigma^M(t)=\texttt{CERTIFIED}\), \(b^M(t)=0\), the state bound is valid, and the realised service satisfies the tracking lower bound. If any of these conditions fail, the operator still returns a logged action and diagnostic report, but the hard guarantee is not claimed for that tick.

\subsection{Contract validity and infeasibility semantics}
\label{sec:semantics_short}
\label{app:fallback}

\paragraph{Envelope breach indicator.}
Let the realised total arrival be \(A^M(t)\triangleq a^M(t)+\mathrm{in}^M(t)\).
Define the per-tick breach indicator as
\begin{equation}
b^M(t)\triangleq \1\{A^M(t)>\bar A^M(t)+\eta_A\},
\label{eq:breach_indicator}
\end{equation}
where \(\eta_A\ge 0\) is a measurement tolerance.
Breaches are recorded explicitly and determine when downstream envelope contracts may be used, as described in Section~\ref{app:contracts}.

\paragraph{Always-defined behaviour under infeasibility.}
Infeasibility is expected under overload or under constraints that are too tight for the available capacity.
A mechanism that becomes undefined under stress is a poor operational bargain.
It is also easy to exploit.
For this reason, the emergency policy is part of the operator semantics.

\paragraph{Emergency policy: minimise barrier slack.}
When the certified constraints are infeasible, the operator solves a relaxed projection with explicit nonnegative slack variables.
Let \(u=(\mu,\delta)\), and let \(\ell_i(t)\) denote the compiled lower bound on \(s_i=\mu_i+\delta_i\) induced by safety and drift constraints.
With slack vector \(\xi\ge 0\), the operator solves
\begin{align}
\min_{\mu,\delta,\xi}\quad &
\|u-\tilde u\|_{W_M}^2 + \alpha\,\mathbf{1}^\top \xi
\label{eq:main_emergency_qp_obj}
\\
\text{s.t.}\quad &
\sum_i \mu_i \le C_M(t),\quad \mu\ge 0,\quad \delta\ge 0,
\label{eq:main_emergency_qp_cap}
\\
&
(\mu_i+\delta_i) + \xi_i \ge \ell_i(t)\quad \forall i.
\label{eq:main_emergency_qp_slack}
\end{align}
If the optimum has \(\xi^\star=0\), the operator reports \texttt{CERTIFIED}.
Otherwise, it reports \texttt{INFEASIBLE} and returns \((u^\star,\xi^\star)\).
Thus behaviour is always defined, and the violation is exported as a quantified and auditable signal.

The slack \(\xi^\star\) certifies how far the system is from satisfying its hard constraints under the current capacity and declared envelopes.
The slack \(\xi^\star\) separates two regimes. When \(\xi^\star=0\), the original guarantees apply on valid certified ticks. When \(\xi^\star>0\), the operator has reached a structurally infeasible point and returns the best-effort action under the stated slack criterion, and produces \texttt{INFEASIBLE} tick. Such \texttt{INFEASIBLE} ticks can arise even with a perfect proposer, for example under very tight \(Q^{\max}\), large envelopes, actuation delay \(\tau_u\), or insufficient resources. We therefore report both the \texttt{INFEASIBLE} fraction and the slack magnitude as diagnostic quantities.

% =========================
\subsection{Certified closed-loop control}
\label{sec:pcp}
% =========================

The proposer, whether learned or heuristic, observes telemetry history \(y^M(\le t)\) and emits a candidate action \(\tilde u^M(t)\).
\begin{equation}
\tilde u^M(t) = \pi_\theta^M(\phi^M(y^M(\le t))).
\label{eq:proposer}
\end{equation}
In the basic parameterisation, \(\tilde u^M(t)=(\tilde\mu^M(t),\tilde\delta^M(t))\).

\paragraph{Certified projection: operator realisation.}

The certified operator returns platform-level targets, written as \((\mu^{\mathrm{tar}}(t),\delta(t))\).
When the platform implements these targets, realised removals satisfy the bounded tracking model in Section~\ref{app:tracking}.

\begin{align}
  \begin{split}
      u^M(t) = \arg\min_{u\in\mathcal{U}_{\mathrm{cert}}^M(t;\Theta^M)} \|u-\tilde u^M(t)\|_{W_M}^2,
\\ \text{when } \sigma^M(t)=\texttt{CERTIFIED}.
\label{eq:projection}
  \end{split}
\end{align}

For bottleneck scheduling with \(\mu^M\) and \(\delta^M\), a practical convex form is the following small quadratic programme:
\begin{align}
\min_{\mu^M,\delta^M}\quad &
\|\mu^M-\tilde\mu^M\|_2^2 + \lambda_M\|\delta^M-\tilde\delta^M\|_2^2 + M_M \mathbf{1}^\top \delta^M
\label{eq:qp_obj}
\\
\text{s.t.}\quad &
\sum_i \mu_i^M \le C_M(t),\quad \mu_i^M\ge 0,\quad \delta_i^M\ge 0,
\label{eq:qp_cap}
\\
&
q_i^M(t) + \bar A_i^M(t) - (\mu_i^M+\delta_i^M) \le Q_{i}^{M,\max}\quad \forall i,
\label{eq:qp_safety}
\\
&
q_i^M(t)\ge Q_{i}^{M,\mathrm{hi}} \Rightarrow (\mu_i^M+\delta_i^M) \ge \bar A_i^M(t)+\varepsilon_i^M\quad \forall i,
\label{eq:qp_drift}
\end{align}
with optional linear constraints such as \eqref{eq:floor} and \eqref{eq:cap_mitig}.
In implementation, the implication in \eqref{eq:qp_drift} is enforced by activating the constraint for every index with \(q_i^M(t)\ge Q_{i}^{M,\mathrm{hi}}\).

\FloatBarrier

\begin{center}
\begin{minipage}{\linewidth}
  \centering
  \includegraphics[width=.8\linewidth]{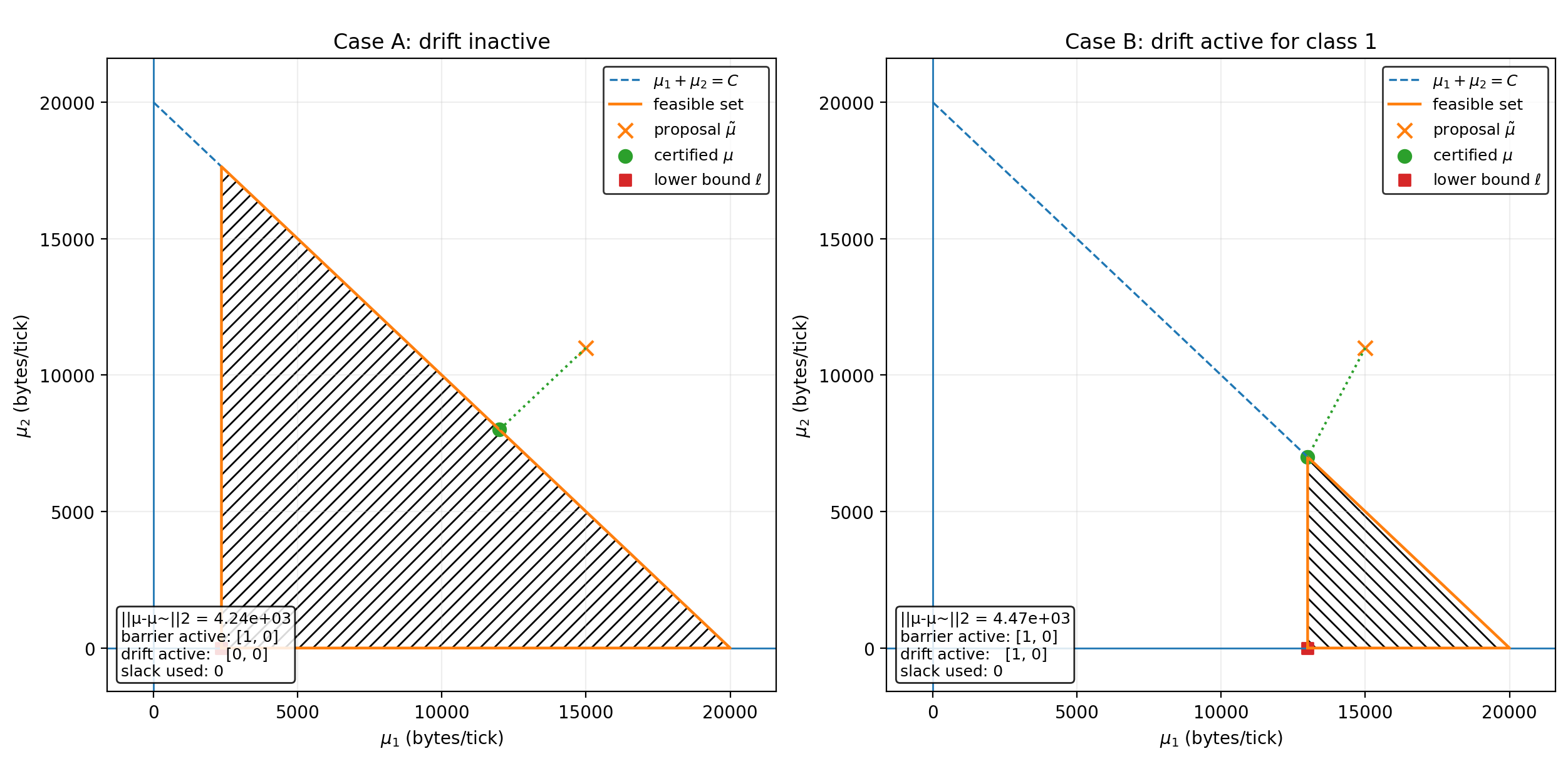}
  \captionof{figure}{%
  \textbf{Certificate-induced projection geometry by class.}
  Left: drift inactive.
  Right: drift active for class~1.
  The hatched polygon shows the feasible set induced by the certificate in the \((\mu_1,\mu_2)\) plane, measured in bytes per tick.
  The set is bounded by the capacity \(C\) and the compiled per-class lower bound \(\ell\).%
  }
  \label{fig:proj_geometry}
\end{minipage}
\end{center}

\FloatBarrier

Figure~\ref{fig:proj_geometry} illustrates the projection view. The certificate induces a convex feasible region in the \((\mu_1,\mu_2)\) plane, and the executed action is the metric projection of the proposal onto that region. When drift is inactive, the feasible set is broad. When drift activates for class~1, the lower bound tightens, the feasible set shrinks, and the projection shifts service toward that class while preserving the capacity constraint. If the feasible set is empty, the operator reports \texttt{INFEASIBLE} and enters the slack-handling mode.

\paragraph{Operator-level safety theorem.}

\begin{theorem}[Safety of the certified operator]
\label{thm:op_safety}
Assume \eqref{eq:envelopes}. Consider a module \(M\) whose operator \(\mathcal{C}_\Theta^M\) reports \(\sigma^M(t)=\texttt{CERTIFIED}\) at time \(t\). Suppose the state used for certification is valid for the current queue state, and the operator enforces the barrier constraint \eqref{eq:barrier_safety} using the declared arrival envelope \(\bar A^M(t)\). If \(q^M(t)\in\mathcal{S}^M\), then \(q^M(t+1)\in\mathcal{S}^M\). Consequently, if \(q^M(0)\in\mathcal{S}^M\), the required envelopes and state bounds remain valid, and \(\sigma^M(t)=\texttt{CERTIFIED}\) for all \(t\), then \(q^M(t)\in\mathcal{S}^M\) for all \(t\ge 0\).
\end{theorem}

% =========================
\section{Delay, partial observability, and algorithms}
\label{sec:delay}
% =========================

\paragraph{Envelope and state-bound construction for deployment.}
\label{sec:env_bounds_recipe}

The operator interface needs two quantities. The first is a declared arrival envelope \(\bar A(t)\). The second is a certified backlog interval,
\[
q(t)\in[\underline q(t),\overline q(t)].
\]
The theorem above treats these quantities as valid inputs. In deployed packet systems, however, they must be constructed from delayed counters, sampled telemetry, and local traffic assumptions. We now describe this construction. In practice, \(\bar A(t)\) is obtained in one of three ways, depending on what the operator is allowed to assume.

\begin{enumerate}
\item \textbf{Policy envelope (contracted).}
For traffic classes with explicit contracts (tenant limits, policers, admission control), the envelope is specified by configuration, for example token-bucket style:
\begin{equation}
\bar A(t) = \min\{B + R\Delta,\ \bar A^{\max}\},
\label{eq:env_tokenbucket}
\end{equation}
where $R$ is a configured rate, $B$ is a burst budget, and $\Delta$ is the control interval.
This is the most direct point of contact with network calculus: a token-bucket contract is an arrival-curve style assumption. In our framework, however, it serves as an input contract to runtime certification rather than as the whole analysis.
\item \textbf{Empirical envelope (calibrated).}
When no explicit contract exists, we calibrate $\bar A(t)$ from a recent history window of measured arrivals, using a high quantile and a slack margin:
\begin{equation}
\bar A(t) \;=\; Q_{p}\!\left(\{A(\tau)\}_{\tau=t-W}^{t-1}\right) + m_A,
\label{eq:env_quantile}
\end{equation}
where $Q_p(\cdot)$ is the empirical $p$-quantile (e.g., $p=0.99$ or $0.999$), $W$ is a window length in ticks, and $m_A$ is a configured safety margin.

\item \textbf{Hybrid envelope (policy plus calibration).}
A common compromise is to cap the empirical envelope by a configured maximum and to floor it by a minimum contract:
\begin{equation}
\bar A(t) \;=\; \min\big\{\bar A^{\max},\ \max\{\bar A^{\min},\ Q_p(\cdot)+m_A\}\big\}.
\label{eq:env_hybrid}
\end{equation}
\end{enumerate}

\paragraph{Backlog bounds from delayed counters and local accounting.}
Let the platform expose a delayed backlog-related signal at time $t$ (a queue occupancy counter, a byte backlog estimate, or an aggregate) of the form
$y(t)\approx q(t-\tau_y)$, possibly with aggregation and noise as in \eqref{eq:telemetry}.
To obtain a bound at the current tick, we use two facts that are available to the controller:
(i) the declared envelope upper-bounds arrivals between counter samples, and
(ii) the controller knows the actions it has already scheduled to apply (including any actuation buffer). A simple upper bound is:
\begin{equation}
\overline q(t)
\;=\;
\Big[y(t) + \sum_{s=0}^{\tau_y-1}\bar A(t-\tau_y+s) - \sum_{k=0}^{\tau_y-1}\hat s(t-\tau_y+k)\Big]_+ + \eta_q,
\label{eq:q_upper_from_delay}
\end{equation}
where $\hat s(\cdot)$ is the realised (or conservatively lower-bounded) total removal scheduled/applied on each tick
and $\eta_q\ge 0$ is a tolerance that absorbs counter noise, timestamp jitter, and aggregation error.

If the platform provides a conservative removal \emph{lower bound} $\hat s_{\min}(\cdot)$ (Section~\ref{app:tracking}),
then \eqref{eq:q_upper_from_delay} remains valid by substituting $\hat s_{\min}$.
If only an upper bound on removal is available, we keep safety conservative by relying on \eqref{eq:q_upper_from_delay} only (upper-bound mode).

A matching lower bound is optional in the core mechanism.
When required (for example, to reduce conservatism in drift-trigger activation), a standard conservative construction is:
\begin{equation}
\underline q(t)
\;=\;
\Big[y(t) - \sum_{k=0}^{\tau_y-1}\hat s_{\max}(t-\tau_y+k)\Big]_+ - \eta_q,
\label{eq:q_lower_from_delay}
\end{equation}
with $\hat s_{\max}$ an upper bound on applied removals. We clip $\underline q(t)$ at $0$.

When only bounds are available, the operator substitutes \(\overline q(t)\) wherever \(q(t)\) appears in safety constraints. Drift activation also uses \(\overline q(t)\) in the conservative mode. Thus the notation is fixed throughout: safety constraints are certified against the upper backlog bound, while \(\underline q(t)\) is used only when a lower bound is explicitly needed.

\paragraph{Applied-action consistent certification under actuation delay (conservative).}
\label{sec:delay_applied}

The operator can be read as a domain-specific safety filter: a proposer supplies a candidate action, and an online projection step enforces safety and feasibility \cite{ames2017cbfqp,ames2019cbf,wabersich2021psf}.

With actuation delay, the operator chooses $u(t)$ now, but the system applies $u(t)$ at $t+\tau_u$.
Meanwhile, actions already placed in the actuation pipeline will be applied on ticks $t,t+1,\dots,t+\tau_u-1$.

\paragraph{Action buffer model.}
Let $\hat s_i(t+k)$ denote the total removal already scheduled to be applied at future tick $t+k$ from actions chosen earlier (this is known to the controller because it is the controller’s own buffer). Let $\overline q_i(t)$ be a valid upper bound on current backlog and let $\bar A_i(t+s)$ be the declared envelope over the delay window.

\paragraph{A conservative apply-time bound.}
A mechanically checkable upper bound on the backlog just before $u(t)$ takes effect is
\begin{equation}
\overline q_i^{\,\mathrm{apply}}(t+\tau_u)
\triangleq
\Big[\overline q_i(t) + \sum_{s=0}^{\tau_u-1}\bar A_i(t+s) - \sum_{k=0}^{\tau_u-1}\hat s_i(t+k)\Big]_+.
\label{eq:q_apply_bound}
\end{equation}

\paragraph{Applied-action barrier constraint.}
To ensure the backlog cap at the first tick when the new action is applied, it is sufficient to enforce
\begin{equation}
\overline q_i^{\,\mathrm{apply}}(t+\tau_u) + \bar A_i(t+\tau_u) - s_i^{\mathrm{exec}}(t) \le Q_i^{\max},
\label{eq:barrier_applied}
\end{equation}
where $s_i^{\mathrm{exec}}(t)=\mu_i(t)+\delta_i(t)$ is the removal implied by the executed action returned by the operator at time $t$.

\paragraph{Applied-action drift trigger.}
A conservative version activates drift using $\overline q_i^{\,\mathrm{apply}}(t+\tau_u)$ and enforces
\begin{equation}
\overline q_i^{\,\mathrm{apply}}(t+\tau_u)\ge Q_i^{\mathrm{hi}}
\ \Rightarrow\
s_i^{\mathrm{exec}}(t)\ge \bar A_i(t+\tau_u)+\varepsilon_i.
\label{eq:drift_applied}
\end{equation}

% =========================
\subsection{Algorithms}
\label{sec:algorithms}
% =========================

Algorithm~\ref{alg:pcp} is the core online loop for one module $M$ (a queueing/scheduling element).
Inputs are deliberately minimal and match what deployments typically have:
(i) telemetry history $y^M(\le t)$, which may be delayed and noisy,
(ii) a backlog bound 
$[\underline q^M(t),\overline q^M(t)]$ (or only $\overline q^M(t)$ in the simplest mode),
(iii) declared envelopes for exogenous and inflow arrivals, which combine into a total envelope $\bar A^M(t)$, and
(iv) a capacity estimate $C_M(t)$ and action limits.
The operator compiles three kinds of certificates into constraints: barrier-style safety caps (Type~B), drift triggers (Type~A), and contract rules that define exported envelopes (Type~C).
It then solves a projection (or uses a closed-form fast path when applicable) to obtain an executed action that is feasible with respect to these constraints.
If the constraints are infeasible, the operator reports this explicitly through $\sigma(t)$ (and, if enabled, returns an emergency action with quantified slack as in Section~\ref{app:fallback}).
The module then applies the executed action (immediately or after  actuation delay) and updates the queue state by \eqref{eq:q_update_module}.

\begin{algorithm}[t]
\caption{Certified operator step at module $M$ (per time tick)}
\label{alg:pcp}
\begin{algorithmic}[1]
\Require Telemetry history $y^M(\le t)$, bounds $[\underline q^M(t),\overline q^M(t)]$,
         exogenous envelope $\bar a^M(t)$, inflow envelope $\bar{\mathrm{in}}^M(t)$,
         capacity $C_M(t)$
\State $\bar A^M(t)\gets \bar a^M(t)+\bar{\mathrm{in}}^M(t)$
\State Proposer outputs $(\tilde\mu^M(t),\tilde\delta^M(t)) \gets \pi_\theta^M(\phi^M(y^M(\le t)))$
\State Compile per-tick certificate constraints from $\Theta^M$ using $\overline q^M(t)$ and $\bar A^M(t)$
\State Solve projection to obtain executed $(\mu^M(t),\delta^M(t))$ and set status $\sigma^M(t)$
\State Export envelopes $\bar z_{M\to V}(t)$ and diagnostics $r^M(t)$
\State Apply action (immediately or with actuation delay) and update queues via \eqref{eq:q_update_module}
\end{algorithmic}
\end{algorithm}

For completeness, we give the DAG propagation, cyclic closure, and stress harness algorithms in Appendix~\ref{app:algorithms}.

\section{Stability guarantees}
\label{sec:stability}
% =========================

The safety condition can be read as a queue-domain analogue of barrier-function feasibility enforced online by a correction step, while the stability condition is deliberately aligned with drift-style arguments used throughout queueing-network control \cite{ames2017cbfqp,ames2019cbf,neely2010stochastic}.
We give a Foster--Lyapunov drift result aligned with the enforced drift trigger. For clarity, we drop the module superscript and write $q_i(t)$, $\bar A_i(t)$, and $s_i(t)$.

\begin{lemma}[One-step quadratic drift bound]
\label{lem:quad_drift}
Let $V(q)=\sum_{i=1}^N q_i^2$. For the update $q_i(t+1)=[q_i(t)+x_i(t)]_+$ with $x_i(t)=A_i(t)-s_i(t)$, it holds that
\begin{equation}
V(q(t+1)) - V(q(t))
\le 2\sum_{i=1}^N q_i(t)\,x_i(t) + \sum_{i=1}^N x_i(t)^2.
\label{eq:drift_basic}
\end{equation}
\end{lemma}

\begin{theorem}[Foster--Lyapunov stability under certified drift trigger]
\label{thm:strong_stability}
Assume $0\le A_i(t)\le \bar A_i(t)\le \bar A_i^{\max}<\infty$ and $0\le s_i(t)\le \bar s_i<\infty$ for all classes $i=1,\dots,N$ and all $t$.
Fix thresholds $Q_i^{\mathrm{hi}}>0$ and margins $\varepsilon_i>0$, and suppose the executed action satisfies
\begin{equation}
q_i(t)\ge Q_i^{\mathrm{hi}} \;\Rightarrow\; s_i(t)\ge \bar A_i(t)+\varepsilon_i,
\quad \forall i,\ \forall t.
\label{eq:dominance_thm}
\end{equation}
Let $\varepsilon_{\min}\triangleq \min_{1\le i\le N} \varepsilon_i$ and $B_0\triangleq \sum_{i=1}^N Q_i^{\mathrm{hi}}$. Define
\begin{equation}
\widetilde K \triangleq \sum_{i=1}^N \bigl[2 Q_i^{\mathrm{hi}} \bar A_i^{\max} + (\bar A_i^{\max}+\bar s_i)^2\bigr],
\qquad
\beta \triangleq 2\varepsilon_{\min} B_0 + \widetilde K.
\label{eq:beta_def}
\end{equation}
Then, for all $t$,
\begin{equation}
\EE\!\left[V(q(t+1)) - V(q(t)) \mid q(t)\right]
\le -2\varepsilon_{\min}\,\|q(t)\|_1 + \beta,
\label{eq:global_drift}
\end{equation}
where $V(q)=\sum_{i=1}^N q_i^2$ and $\|q(t)\|_1=\sum_{i=1}^N q_i(t)$.
In particular,
\begin{equation}
\limsup_{T\to\infty}\frac{1}{T}\sum_{t=0}^{T-1}\EE[\|q(t)\|_1]
\;\le\; \frac{\beta}{2\varepsilon_{\min}}.
\label{eq:avg_backlog}
\end{equation}
\end{theorem}

\paragraph{Alignment with the operator.}
Condition \eqref{eq:dominance_thm} is exactly the drift trigger \eqref{eq:drift_trigger} enforced when $\sigma(t)=\texttt{CERTIFIED}$ with $A_i=\bar A_i$. Thus the theorem applies to executed actions by construction, independent of proposer quality.

\section{Envelope contracts and compositional guarantees}
\label{app:contracts}
% -------------------------

\subsection{Contract semantics: how $z$ and $\bar z$ are computed}
\label{sec:contract_semantics}

We make the contract signals concrete so that composition is mechanically checkable.

\paragraph{Routing model.}
Fix a routing specification for each edge $U\to V$. In the simplest (and common) case, each class $i$ at $U$ forwards a fixed fraction $r_{U\to V,i}\in[0,1]$ of its removals to $V$, with $\sum_{V} r_{U\to V,i}\le 1$ (the remainder exits the modelled network). More generally, this is a nonnegative class-mapping matrix, but the fraction model is enough to state the idea.

\paragraph{Realised flow signal.}
Let $s_i^U(t)=\mu_i^U(t)+\delta_i^U(t)$ denote the realised total removal from queue $i$ at $U$ on tick $t$ (that is, removal that actually occurs on that tick, after any actuation delay is applied). The realised inter-module flow is
\begin{equation}
z_{U\to V,i}(t)\triangleq r_{U\to V,i}\,s_i^U(t).
\label{eq:z_from_s}
\end{equation}
Stacking across classes gives $z_{U\to V}(t)\in\RR_{\ge 0}^{N}$.

\paragraph{Exported envelope.}
The exported envelope must upper bound what can be sent downstream under the assumptions used for certification. A conservative and easy-to-check choice is
\begin{align}
\begin{split}
    \bar z_{U\to V,i}(t)\triangleq r_{U\to V,i}\,\bar s_i^U(t),
\\
\bar s_i^U(t)\triangleq \min\big(s_i^{U,\mathrm{target}}(t),\ \overline q_i^U(t)+\bar A_i^U(t)\big),
\label{eq:zbar_cap}
\end{split}
\end{align}
where $s_i^{U,\mathrm{target}}(t)$ is the target removal implied by the applied action on tick $t$, and $\overline q_i^U(t)$ is the backlog upper bound used in certification. The cap by $\overline q_i^U+\bar A_i^U$ encodes the physical fact that, within one tick, you cannot remove more work than backlog plus arrivals under the declared envelope.

\paragraph{Why this form.}
With \eqref{eq:zbar_cap}, \(z_{U\to V}(t)\le \bar z_{U\to V}(t)\) follows from
fixed nonnegative routing fractions and the physical bound
\(s_i^U(t)\le \overline q_i^U(t)+\bar A_i^U(t)\). Composition therefore relies
only on exported envelopes, provided the declared arrival contract is respected.

\FloatBarrier

\begin{center}
\begin{minipage}{\linewidth}
  \centering
  \includegraphics[width=\textwidth]{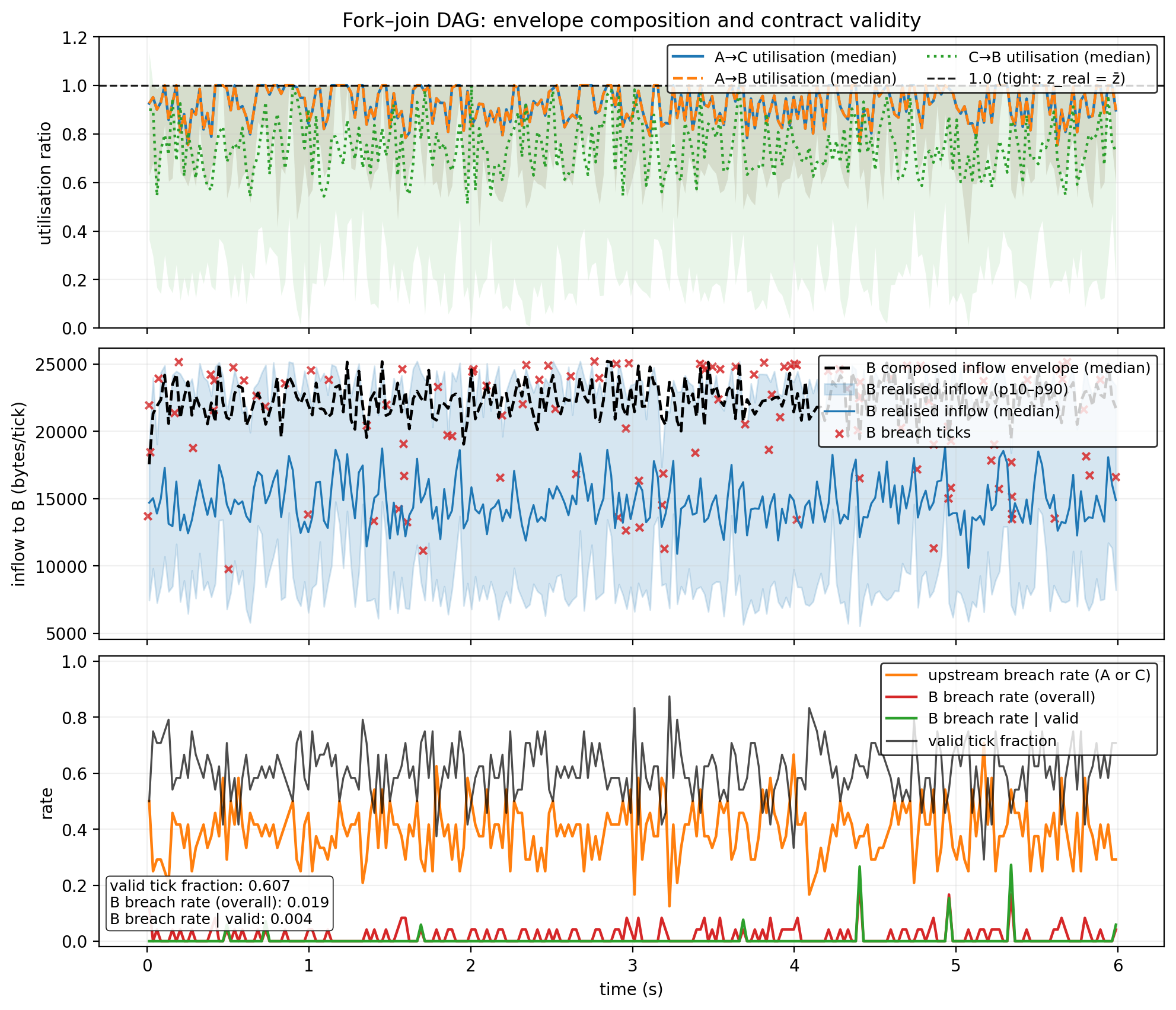}
  \captionof{figure}{%
  \textbf{Fork--join DAG composition, composed envelopes, and contract validity.}
  \textbf{Top:} link utilisation ratios for the fork--join paths, with the tight boundary at utilisation \(=1\).
  \textbf{Middle:} inflow to join node \(B\): composed inflow envelope at \(B\) ( dashed black) versus realised inflow (median and p10--p90 band), with join breach ticks marked.
  \textbf{Bottom:} breach and validity rates over time: upstream breach rate (from \(A\) or \(C\)), join breach rate at \(B\) (overall and conditioned on valid ticks), and the valid tick fraction.%
  }
  \label{fig:forkjoin}
\end{minipage}
\end{center}

\FloatBarrier

\subsection{Breach detection and semantics}
The guarantees in Sections~\ref{sec:stability} and \ref{sec:dag} are conditional
on the declared arrival envelope. If
\(A^M(t)>\bar A^M(t)+\eta_A\), the operator records a breach, and downstream
modules may rely on exported envelopes only on valid ticks. Pessimistic envelopes
remain safe but conservative. Optimistic envelopes increase breach and
infeasibility rates.

% -------------------------

\subsection{Compositional guarantees}
% =========================

\paragraph{DAG networks.}
\label{sec:dag}

We use a three-node fork--join DAG example, with nodes $A,C,B$ and edges $A\to C$, $A\to B$, $C\to B$. Node $A$ forks traffic; node $B$ is the join. In topological order $A\prec C\prec B$, envelope propagation is one pass:
\begin{equation}
\bar{\mathrm{in}}^C(t)\ge \bar z_{A\to C}(t),\qquad
\bar{\mathrm{in}}^B(t)\ge \bar z_{A\to B}(t)+\bar z_{C\to B}(t).
\label{eq:forkjoin_env}
\end{equation}

\begin{theorem}[Compositional safety and stability for a DAG]
\label{thm:dag_comp}
Consider a network of modules connected as a DAG. Suppose each module enforces (i) the safety barrier \eqref{eq:barrier_safety} and (ii) the drift trigger \eqref{eq:drift_trigger}, both with respect to its total envelope $\bar A^M=\bar a^M+\bar{\mathrm{in}}^M$.
Suppose inflow envelopes are assigned by \eqref{eq:compose_in} using exported outflow envelopes \eqref{eq:export_env} in a topological order. Then:
\begin{itemize}
\item (Safety) if $q^M(0)\in\mathcal{S}^M$ for all modules, then $q^M(t)\in\mathcal{S}^M$ for all $t$.
\item (Stability) each module satisfies the Foster--Lyapunov drift bound \eqref{eq:global_drift} with module-specific constants, hence has a finite long-run average backlog bound as in \eqref{eq:avg_backlog}.
\end{itemize}
\end{theorem}

\paragraph{Operator phrasing.}
Equivalently, if every module runs $\mathcal{C}_{\Theta^M}^M$ and reports $\sigma^M(t)=\texttt{CERTIFIED}$ for all $t$, then safety and stability follow module-wise, and exported envelopes suffice for composition.
Figure~\ref{fig:forkjoin} instantiates the fork--join DAG in Section~\ref{sec:dag} and visualises the contract semantics behind \eqref{eq:compose_in}.
At the join node \(B\), the composed inflow envelope (dashed black) is formed by summing upstream exported envelopes, and it typically upper-bounds the realised inflow (median with p10--p90 band), even in a near-saturation regime (top panel).
We do not treat breaches as harmless noise. Upstream envelope violations reduce the valid-tick fraction, and downstream contracts may be used only on valid ticks. Consistent with this semantics, the breach rate at \(B\), conditioned on valid ticks, is near zero. The remaining breaches occur mainly on ticks already invalidated by upstream breach flags, as shown in the bottom panel. This is the intended fail-loud behaviour: envelope composition supports modular reasoning while its assumptions hold, and withdraws that guarantee when they fail.

\paragraph{Cyclic networks.}
\label{sec:cycles}

DAGs avoid circular dependencies in envelope assignment. Cyclic networks require an additional closure condition, since each module's assumptions may depend on guarantees exported by other modules.

\paragraph{Envelope closure problem.}
Stack all inflow envelopes in \(\bar{\bm{\mathrm{in}}}(t)\in\RR_{\ge 0}^d\), and all exported outflow envelopes in \(\bar{\bm z}(t)\in\RR_{\ge 0}^d\).

The interconnection mapping $G$ converts outflow envelopes to inflow envelopes:
\begin{equation}
\bar{\bm{\mathrm{in}}}(t) = G\big(\bar{\bm z}(t)\big),
\label{eq:G_map}
\end{equation}
which is typically linear and monotone (routing fractions and class mapping).

Each module produces export envelopes as a function of its assumed total envelope, which includes exogenous envelope $\bar{\bm a}(t)$ and inflow envelope:
\begin{equation}
\bar{\bm z}(t) = F\big(\bar{\bm a}(t) + \bar{\bm{\mathrm{in}}}(t)\big).
\label{eq:F_map}
\end{equation}
Closure is the fixed point
\begin{equation}
\bar{\bm{\mathrm{in}}}(t) = \Phi_t\big(\bar{\bm{\mathrm{in}}}(t)\big),
\qquad
\Phi_t(x) \triangleq G\!\left(F\big(\bar{\bm a}(t)+x\big)\right).
\label{eq:Phi}
\end{equation}

\begin{definition}[Monotone Lipschitz envelope maps]
A map $\Psi:\RR_{\ge 0}^d\to\RR_{\ge 0}^d$ is monotone if $x\le y$ implies $\Psi(x)\le \Psi(y)$ componentwise.
It is Lipschitz with constant $L$ in norm $\|\cdot\|$ if $\|\Psi(x)-\Psi(y)\|\le L\|x-y\|$ for all $x,y$.
\end{definition}

\begin{theorem}[Cyclic envelope closure under a small-gain condition]
\label{thm:closure}
Fix time $t$ and suppose: (i) $G$ is monotone and Lipschitz with constant $L_G$, (ii) $F$ is monotone and Lipschitz with constant $L_F$, and (iii) $L_G L_F < 1$.
Then $\Phi_t$ in \eqref{eq:Phi} is a contraction and has a unique fixed point $\bar{\bm{\mathrm{in}}}^\star(t)$.
Moreover, the iteration
\begin{equation}
\bar{\bm{\mathrm{in}}}^{(k+1)}(t) \leftarrow \Phi_t\big(\bar{\bm{\mathrm{in}}}^{(k)}(t)\big),\qquad k=0,1,2,\dots
\label{eq:closure_iter}
\end{equation}
converges to $\bar{\bm{\mathrm{in}}}^\star(t)$ from any initial $\bar{\bm{\mathrm{in}}}^{(0)}(t)\in\RR_{\ge 0}^d$ at a linear rate bounded by $L_G L_F$.
\end{theorem}

\paragraph{Operational test.}
In practice, we use a conservative small-gain diagnostic based on a nonnegative gain matrix upper bound for $F$ combined with routing $G$. The spectral and norm-based checks, and the cycle-closure phase-transition plot, are in Section~\ref{app:cycles}.

% =========================

% -------------------------

\subsection{Cyclic closure diagnostics and conservative small-gain tests}
\label{app:cycles}
% -------------------------

In many packet models, $G$ is linear: $\bar{\bm{\mathrm{in}}}=R\bar{\bm z}$ with $R\ge 0$ encoding routing fractions and class mapping.
If $F$ can be upper bounded by an affine monotone map $\bar{\bm z}\le S(\bar{\bm a}+x)$ with $S\ge 0$, then $\Phi_t(x)\le RS(\bar{\bm a}+x)$ and a sufficient small-gain condition is
\begin{equation}
\rhoSpec(RS)<1,
\label{eq:spectral_smg}
\end{equation}
where $\rhoSpec(\cdot)$ is the spectral radius.

A simple instantiation is an affine monotone export rule where each component of $\bar{\bm z}$ is a nonnegative linear function of the assumed total envelope $\bar{\bm a}+x$:
\begin{equation}
F(\bar{\bm a}+x) \triangleq S(\bar{\bm a}+x),
\label{eq:F_linear}
\end{equation}
where $S\ge 0$ is a gain matrix.

In the induced $1$-norm, a conservative Lipschitz constant is
\begin{equation}
L_F \le \|S\|_1 = \max_{j}\sum_i S_{ij},
\label{eq:LF_bound}
\end{equation}
The bound is monotone because \(S\) is nonnegative.
When \(G(\bar{\bm z})=R\bar{\bm z}\) with \(R\ge 0\), we have \(L_G\le \|R\|_1\).
A sufficient condition for contraction is \(\|R\|_1\|S\|_1<1\). This test is conservative, but computationally inexpensive enough to be verified online.

\paragraph{Worked cyclic interpretation.}
Consider two modules with recirculating traffic. The routing matrix \(R\) is obtained from configured traffic fractions, for example the fraction of class \(i\) leaving module \(U\) and returning to module \(V\). The export-gain matrix \(S\) upper-bounds how much a certified outflow envelope can increase when the assumed inflow envelope is enlarged. At a control tick, the operator forms the conservative product \(RS\) and checks \(\rhoSpec(RS)\). If \(\rhoSpec(RS)<1\), the cyclic envelope closure is valid for that tick. If not, the operator can either enlarge envelopes conservatively, mark the closure as invalid, or fall back to a non-compositional local certificate.

Figure~\ref{fig:closure_phase} illustrates the closure test around \(\rhoSpec(RS)=1\). Below the boundary, the fixed-point iteration remains bounded and converges with a small residual. Near and above the boundary, the residual and iteration count increase, and the coupled queue-level view shows the corresponding rise in overload pressure and p99 delay.

\FloatBarrier

\begin{center}
\begin{minipage}{\linewidth}
  \centering
  \includegraphics[width=.63\linewidth]{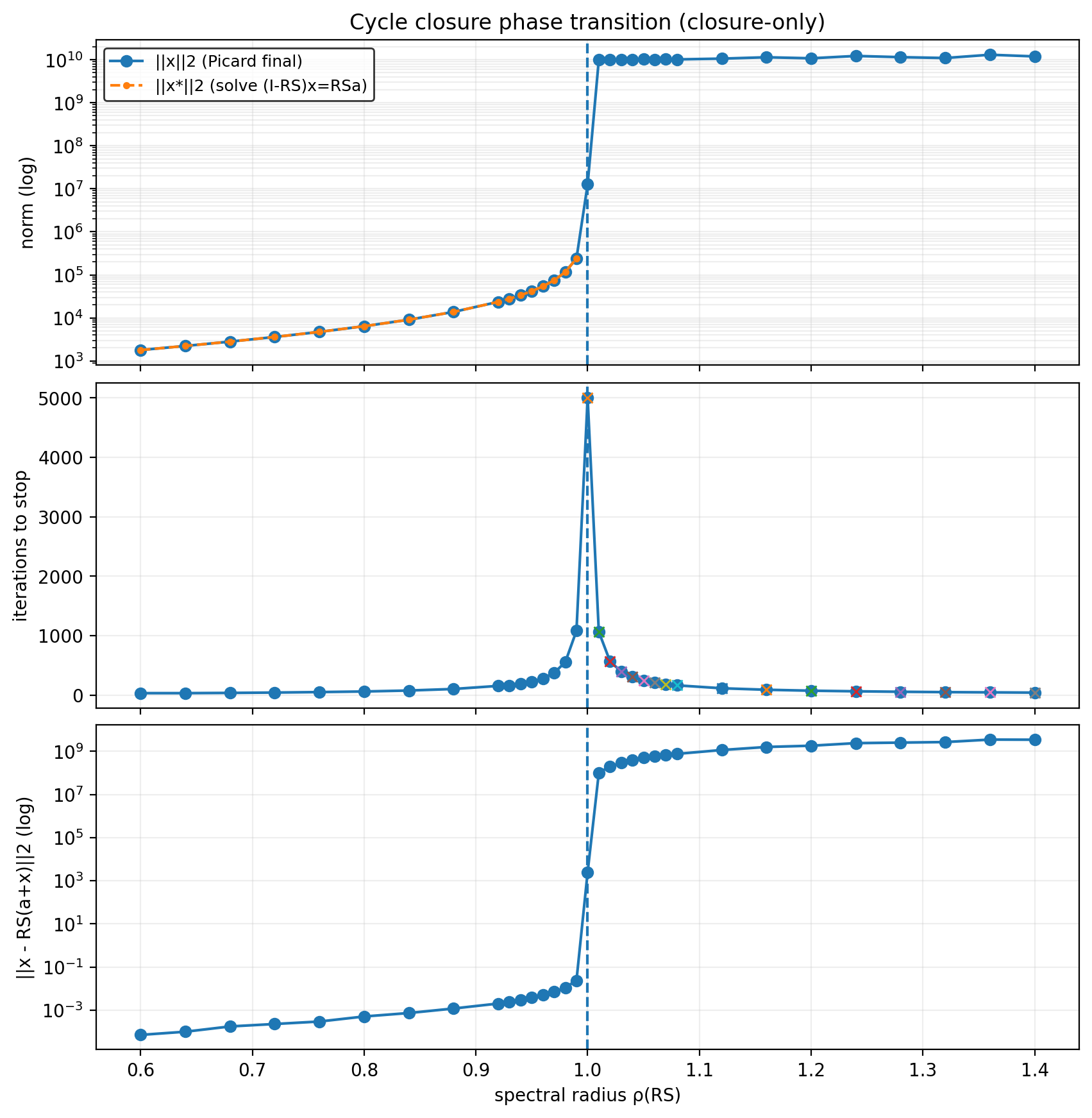}

  \vspace{0.8ex}
  {\small\textbf{(a)} Closure-only diagnostics versus spectral radius \(\rhoSpec(RS)\).}

  \vspace{0.8ex}
  \includegraphics[width=.7\linewidth]{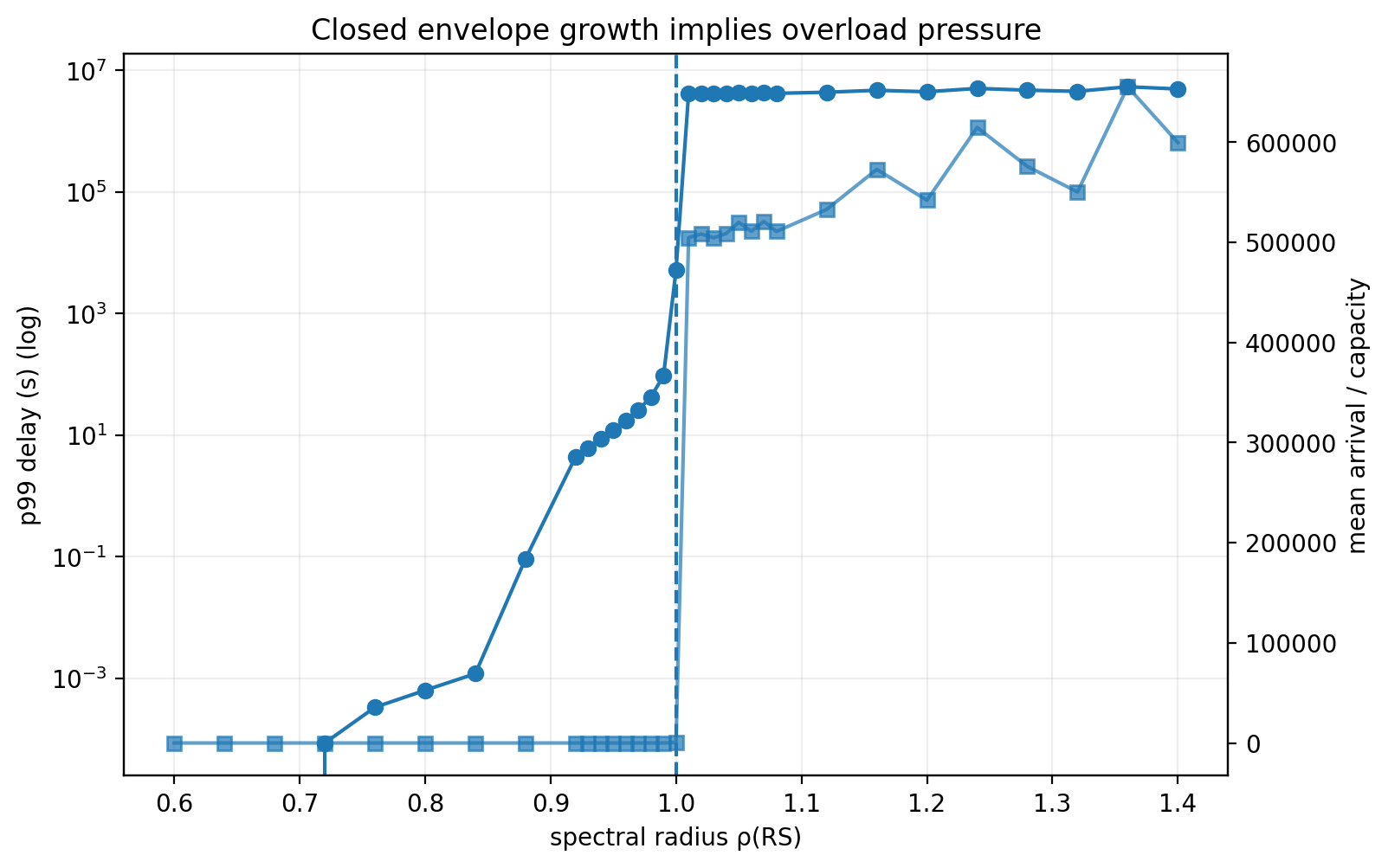}

  \vspace{0.8ex}
  {\small\textbf{(b)} Queue-level overload proxy under the same \(\rhoSpec(RS)\) sweep.}

  \captionof{figure}{%
  \textbf{Cycle closure phase transition around the small-gain boundary.}
  We sweep \(\rhoSpec(RS)\), with the dashed vertical line marking \(\rhoSpec(RS)=1\).
  \textbf{(a)} Closure-only diagnostics for Algorithm~\ref{alg:closure}: Picard iterate norm \(\|x\|_2\) (and \(\|x^\star\|_2\) when defined), iterations-to-stop, and residual \(\|x-RS(a+x)\|_2\).
  \textbf{(b)} Coupled queue-level view: growth in the closed-envelope regime aligns with increased overload pressure and sharply increasing p99 delay.%
  }
  \label{fig:closure_phase}
\end{minipage}
\end{center}

\FloatBarrier
% -------------------------

\section{Practical integration, adoption, and design trade-offs}
\label{sec:practical}
% =========================

This section maps the operator interface onto common packet-system deployments.
The operator sits between a proposer and the execution surface. It may run
locally at a host, gateway, or switch controller for latency-sensitive loops, or
hierarchically, where a central controller proposes actions and each device
certifies them against local constraints. In both cases, correctness attaches to
the action actually executed at the device.

\paragraph{Mapping actions to packet mechanisms.}

The model action \(u(t)=(\mu(t),\delta(t))\) is a canonical representation of two families of controls:
service allocation and shedding/mitigation.
Concrete mappings include:

\paragraph{Scheduling and WFQ/DRR weight control.}
For an egress scheduler with capacity $C(t)$, a common actuation surface is a set of class weights $w_i(t)$.
Given weights, a conservative per-tick service target can be derived as
\begin{equation}
\mu_i^{\mathrm{tar}}(t) = \frac{w_i(t)}{\sum_j w_j(t)}\,C(t),
\label{eq:weight_to_mu}
\end{equation}
and certification can be performed directly in $(\mu,\delta)$ space, then compiled back to weights (Section~\ref{app:adoptability}).
Service floors and caps \eqref{eq:floor}--\eqref{eq:cap_mitig} map naturally to minimum and maximum weight constraints.

\paragraph{Rate limiting and pacing.}
At hosts and gateways, actuation often takes the form of per-class rate caps or pacing rates.
Here $\mu_i^{\mathrm{tar}}(t)$ is translated into a rate target over the next interval, and $\delta_i(t)$ corresponds to policing or admission denial.

\paragraph{AQM and mitigation knobs.}
For AQM, the operator can treat $\delta_i(t)$ as a per-tick drop budget (or marking budget) for each class,
and can enforce mitigation caps to prevent overreaction under false positives or transient bursts.

In all cases, the operator certifies a platform-level target and accounts for
platform tracking through the conservative factor in Section~\ref{app:tracking}, so that guarantees attach to the realised removal used in the queue update.
\paragraph{What must be measured and configured.}
The operator requires a backlog-related signal, a capacity estimate, an arrival envelope, and action limits. Delayed counters are converted into backlog bounds as in Section~\ref{sec:env_bounds_recipe}. The arrival envelope is provided as a policy contract, calibrated from recent history, or set by the hybrid method described earlier. Envelope mismatch is surfaced through breach indicators
rather than hidden.

\paragraph{Certificate configuration.}
The certificate configuration $\Theta$ specifies backlog caps $Q^{\max}$, drift thresholds $Q^{\mathrm{hi}}$ and margins $\varepsilon$,
and any linear constraints such as floors/caps.
These are operational policy objects, not proposer internals.

\paragraph{Runtime path and overhead.}
On each tick, the runtime path is:

\begin{enumerate}
\item update envelopes $\bar A(t)$ and backlog bounds $[\underline q(t),\overline q(t)]$,
\item obtain the proposal $\tilde u(t)$,
\item compile per-tick constraints from $\Theta$ (including delay-consistent constraints when $\tau_u>0$),
\item compute the executed action by projection \eqref{eq:projection} (often via a small QP as in \eqref{eq:qp_obj}--\eqref{eq:qp_drift}),
\item export envelope bounds and emit a status flag with diagnostics.
\end{enumerate}

The projection problem is small, structured, and admits fast paths in common cases (Section~\ref{app:adoptability}).

\FloatBarrier

\begin{center}
\begin{minipage}{\linewidth}
  \centering
  \includegraphics[width=\textwidth]{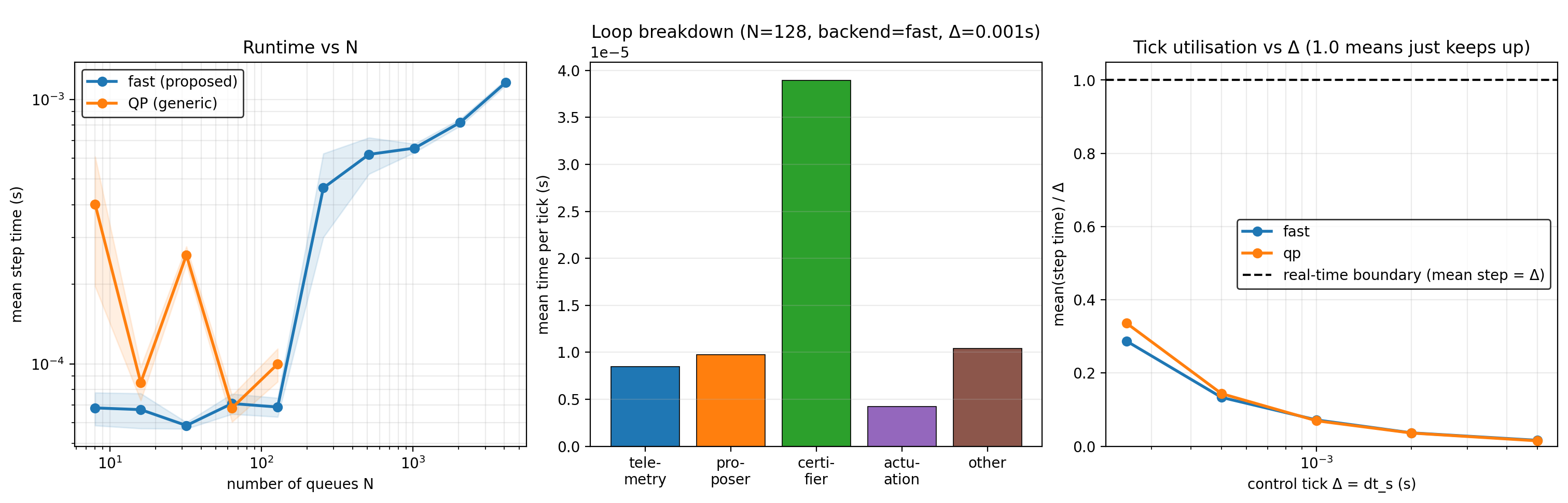}
  \captionof{figure}{\textbf{ Runtime and loop-latency headroom at pinned control interval.}
  \textbf{Left:} Mean certification (operator) step time versus number of queues \(N\), comparing the proposed fast backend against a generic QP baseline.
  \textbf{Middle:} Per-tick loop breakdown at \(N=128\) with fast backend and pinned \(\Delta=1\,\mathrm{ms}\).
  \textbf{Right:} Tick utilisation, across control intervals \(\Delta\).}
  \label{fig:runtime}
\end{minipage}
\end{center}

\FloatBarrier

\begin{table}[t]
\centering
\small
\caption{Per-tick certified-operator runtime at \(\Delta=1\) ms. The table reports certification time only; proposer computation is excluded.}
\label{tab:certifier_runtime}
\begin{tabular}{rccc}
\toprule
\(N\) & mean \(t_{\mathrm{cert}}\) (\(\mu\)s) & p95 \(t_{\mathrm{cert}}\) (\(\mu\)s) & mean/\(\Delta\) (\%) \\
\midrule
64  & 38.5  & 52.0  & 3.8  \\
128 & 37.2  & 50.0  & 3.7  \\
512 & 335.0 & 511.5 & 33.5 \\
\bottomrule
\end{tabular}
\end{table}

\paragraph{Runtime reporting scope.}
The numbers in Table~\ref{tab:certifier_runtime} report the certification step only; proposer execution, telemetry parsing, and platform actuation are excluded from the table. The implementation is the Python \texttt{certloop} package. The fast backend uses a NumPy implementation of projection onto capped simplex constraints with lower bounds, while the generic baseline uses an SLSQP-style constrained optimisation path when the fast structure is not available. The table is averaged over ten seeds in a steady workload at \(\Delta=1\,\mathrm{ms}\), with warm-up excluded. The original run metadata did not record CPU model, kernel, or host scheduling state, so these timings should be read as preliminary certification-only timings rather than a full end-to-end systems benchmark.
Figure~\ref{fig:runtime} reports pinned-tick microbenchmarks for the certification loop.
The left panel reports mean certified-operator step time against the number of queues \(N\). It compares the proposed fast backend with a generic QP baseline.
The fast path is lower throughout and scales smoothly in the regime relevant to per-class scheduling. The middle panel decomposes the per-tick loop at \(N=128\) and \(\Delta=1\,\mathrm{ms}\). The right panel reports tick utilisation, with \(\EE[\text{step time}]/\Delta\) summarised in Table~\ref{tab:certifier_runtime}. These results support the narrower claim that the certification computation itself is small enough for millisecond-scale control intervals under the measured Python backend.

% =========================
\subsection{Adoptability notes and fast certification paths}
\label{app:adoptability}
% -------------------------

The generic operator solves the convex projection \eqref{eq:projection}.
In the default case, which is also the one most likely to appear in practice, this projection has a closed-form structure.

Consider a single bottleneck with constraints \(\mu_i\ge 0\) and \(\sum_i \mu_i \le C\), together with per-class lower bounds on total removal \(s_i=\mu_i+\delta_i\ge \ell_i\).
When shedding is allowed but penalised, the operator usually sets \(\delta=0\) whenever this is feasible.
The executed service then reduces to
\begin{equation}
\mu \;=\; \Pi_{\Delta(C)}(\tilde\mu)\ \ \text{subject to}\ \ \mu_i \ge \ell_i,
\label{eq:fastpath_projection}
\end{equation}
where \(\Pi_{\Delta(C)}\) denotes projection onto the capped simplex.
This can be implemented in \(O(N\log N)\) time by sorting.
The QP backend is needed only when extra constraints, such as caps, floors, or multi-resource coupling, remove this simple structure.

To reduce tuning burden, we recommend three presets for \((Q^{\max},Q^{\mathrm{hi}},\varepsilon)\): latency-biased, throughput-biased, and default.
A simple calibration links backlog caps to delay targets through \(Q^{\max}\approx C d_{\max}\), using consistent units. Then \(Q^{\mathrm{hi}}\) is set as a fraction of \(Q^{\max}\), while \(\varepsilon\) controls recovery speed after bursts.

% -------------------------
\subsection{Positioning and design trade-offs}
\label{sec:discussion}
% =========================

The certified operator creates an explicit execution boundary. The cost is conservatism when envelopes are loose, state bounds are delayed, or \(\kappa_{\min}\) is small. With well-calibrated envelopes, the operator can be less wasteful than static provisioning because it reacts to the certified state. With stale telemetry or pessimistic contracts, it has less room to preserve the proposer’s preferred action. This is the desired behaviour: weaker evidence should produce more cautious execution.

\paragraph{Cost of certification.}
\label{sec:cost_cert}

\FloatBarrier

\begin{center}
\begin{minipage}{\linewidth}
  \centering
  \includegraphics[width=\linewidth]{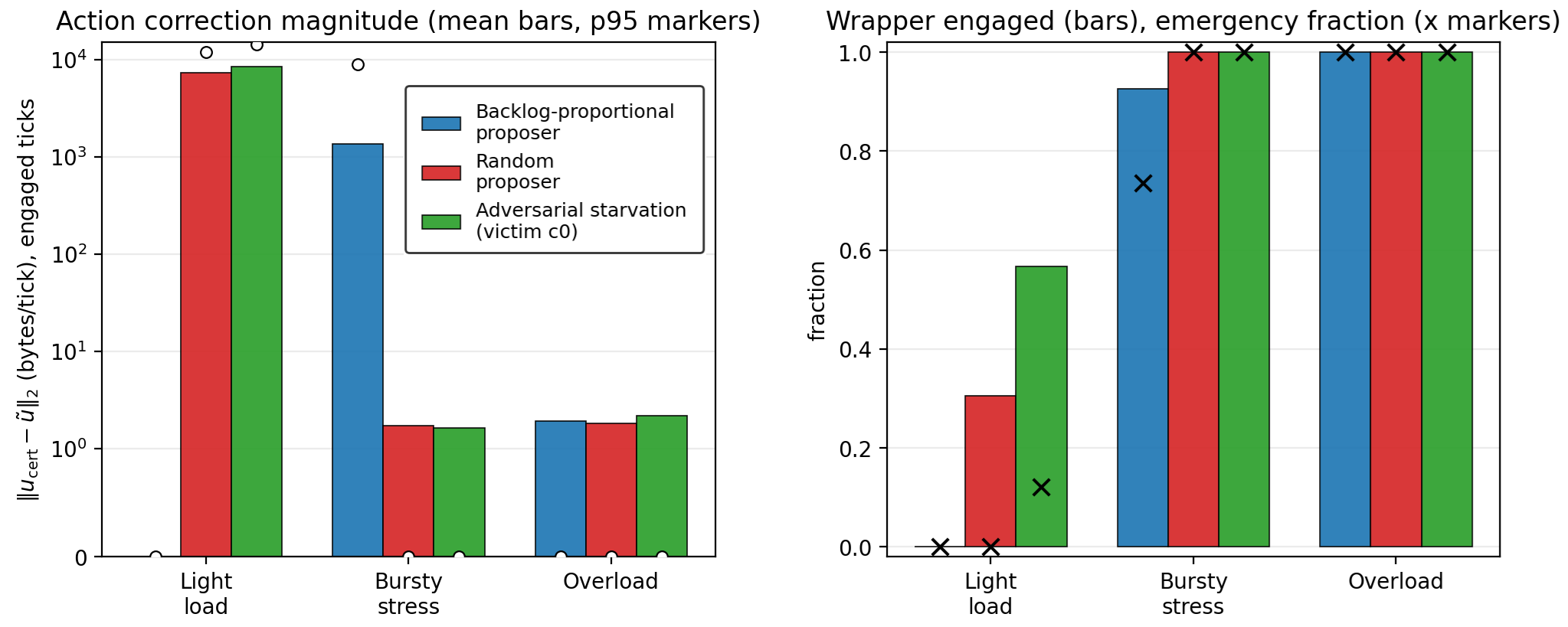}
\captionof{figure}{%
\textbf{Price of certification.}
Top: per-tick projection distance \(\|u(t)-\tilde u(t)\|_2\).
Bottom: fraction of ticks with certified correction and emergency-mode execution.%
}
  \label{fig:price_cert}
\end{minipage}
\end{center}

\FloatBarrier

Figure~\ref{fig:price_cert} reports the projection distance \(\|u(t)-\tilde u(t)\|_2\), separating feasible correction from emergency mode. Under light load, corrections remain small for well-behaved proposers. Under bursty stress, corrections grow as barrier and drift constraints become active. The emergency fraction identifies ticks where the declared envelopes and available capacity make certification infeasible.

% =========================
\section{Experiments}
\label{sec:experiments}
% =========================

The experiments examine safety under stress, correction under weak proposers,
composition across modules, and the effect of cyclic closure on feasibility and
conservatism.

We compare four controllers:
\begin{enumerate}
\item \textbf{Proposer only:} execute \(\tilde u\) directly.
\item \textbf{Monitor only:} execute \(\tilde u\) and raise alarms when constraints are violated.
\item \textbf{Naive clipping:} clip weights, floors, or caps using a heuristic rule, without solving the projection.
\item \textbf{Certified operator (ours):} execute the output of \(\mathcal{C}_\Theta\) and record the resulting status flags.
\end{enumerate}
The proposer may be backlog-proportional, random, or learned. The safety
guarantees apply only on ticks where the certified operator reports
\texttt{CERTIFIED}.

\paragraph{Python execution audit and scope of the backend.}
The present evaluation is deliberately Python based. The simulator is a byte-level closed-loop backend: the proposer emits \(\tilde\mu\), the certified operator emits \(\mu\), actuation delay is represented by a FIFO action buffer, and the realised removal is computed as the physical service that can be taken from backlog plus arrivals on that tick. Thus the experiments audit the execution boundary and the status semantics of the certified operator. They do not claim to measure the short-window tracking loss of Linux \texttt{tc}, or kernel packet scheduling. That loss is the role of the calibrated \(\kappa_{\min}\) parameter in Section~\ref{app:tracking}.

\begin{table}[t]
\centering
\small
\caption{Execution-audit fields used in experiments.}
\label{tab:python_execution_audit}
\begin{tabular}{p{0.20\linewidth}p{0.24\linewidth}p{0.28\linewidth}p{0.18\linewidth}}
\toprule
Trace field & Meaning in the implementation & Role in the paper claim & Limitation \\
\midrule
\(\tilde\mu(t)\), \(\tilde\delta(t)\) & proposer output before certification & shows the action that would have reached the dataplane without the operator & not a safe action by itself \\
\(\mu(t)\), \(\delta(t)\) & certified action returned by the operator & object to which safety, drift, and contract claims attach & still a target for a real scheduler \\
\(s_{\mathrm{real}}(t)\) & realised byte removal in the Python queue update & checks that the simulator executes the logged action subject to available backlog and arrivals & no Linux packetisation or NIC batching \\
\(\sigma(t)\), \(\xi(t)\) & status flag and slack vector & separates certified ticks from infeasible ticks and quantifies violation under overload & guarantees are not claimed when slack is positive \\
\(b(t)\) & envelope-breach indicator & marks ticks where the declared contract is not valid for composition & does not identify the external cause of breach \\
\bottomrule
\end{tabular}
\end{table}

The conservative reading is therefore: the experiments test whether the certified operator enforces the compiled constraints before execution, whether breaches and infeasibility are surfaced, and whether the closed loop remains well defined under delayed telemetry and actuation. Thus, the present evaluation validates the certified execution boundary in a byte-level closed-loop backend.

\paragraph{Stress traces under delayed telemetry and actuation.}
Figure~\ref{fig:killer_traces} evaluates delayed telemetry and actuation. Direct execution and naive clipping react late, producing earlier and more persistent threshold excursions. The certified operator reduces violations while constraints remain feasible by compiling barrier and drift constraints over the available state and envelope bounds. The status panel separates certified ticks from infeasible ticks, so overload is reported rather than hidden.

\FloatBarrier

\begin{center}
\begin{minipage}{\linewidth}
  \centering
  \includegraphics[width=.9\linewidth]{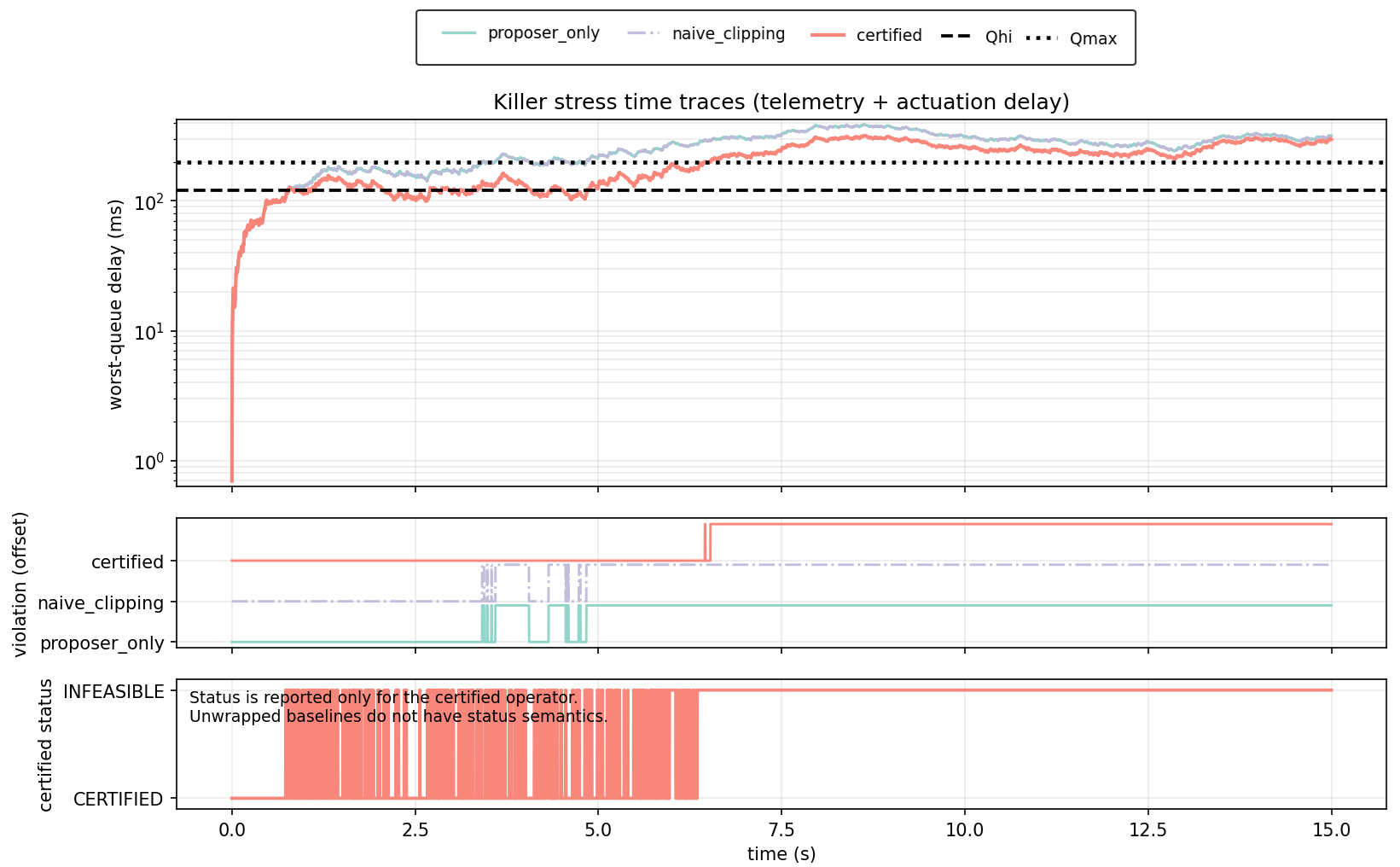}
  \captionof{figure}{%
  \textbf{Killer stress traces under telemetry and actuation delay.}
  }
  \label{fig:killer_traces}
\end{minipage}
\end{center}

\FloatBarrier

\FloatBarrier

\begin{center}
\begin{minipage}{\linewidth}
  \centering
  \includegraphics[width=.97\linewidth]{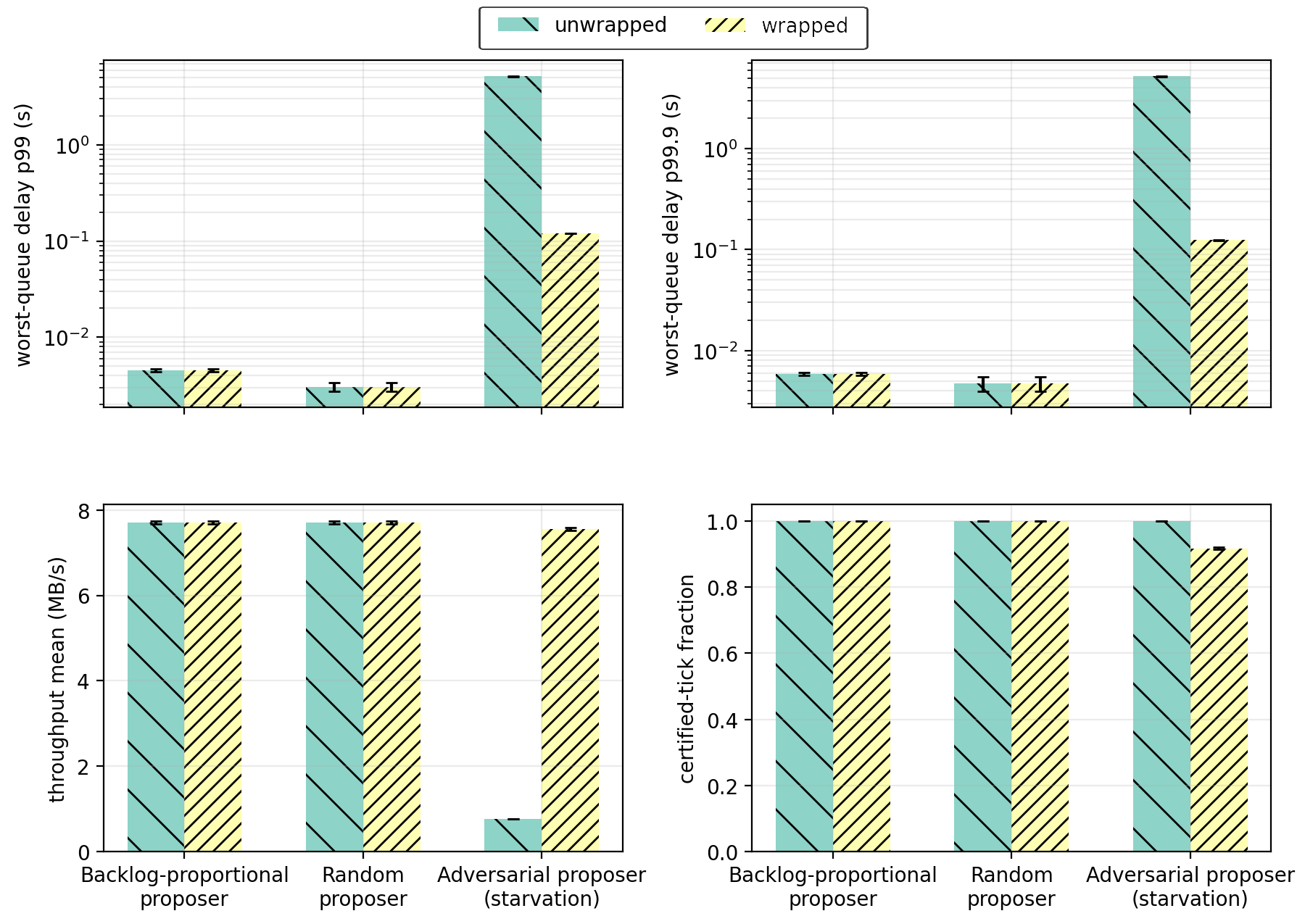}
  \captionof{figure}{%
  \textbf{Certified-operator value across proposer qualities, reported as mean \(\pm\) standard deviation over seeds.}
  Direct execution of proposals is compared with execution through the certified operator for three proposers: backlog-proportional, random, and adversarial starvation.
  Top row: worst-queue delay tails, p99 and p99.9, in seconds on a log scale.
  Bottom-left: mean throughput in MB/s.
  Bottom-right: fraction of ticks on which the wrapped operator reports \texttt{CERTIFIED}.%
  }
  \label{fig:wrapper_value}
\end{minipage}
\end{center}

\FloatBarrier

\paragraph{Certified-operator value across proposer qualities.}
Figure~\ref{fig:wrapper_value} measures the effect of placing the certified operator between a fixed proposer and the dataplane.
For benign proposers, namely backlog-proportional and random, wrapping has little effect on tail delay or throughput.
This is the expected result: their actions already lie close to the certified feasible set, so the operator mostly passes them through.
For the adversarial starvation proposer, direct execution produces much larger worst-queue delay tails and a sharp loss of delivered throughput.
Executing the certified output instead reduces p99 and p99.9 delay and restores throughput, since the executed action must satisfy the barrier and drift-derived lower bounds whenever these constraints are feasible.
The \texttt{CERTIFIED} fraction in the bottom-right panel gives the relevant semantics.
Guarantees apply only on feasible certified ticks, and the operator reports the remaining cases instead of preserving a fiction of correctness.
Across identical seeds, the certified operator consistently reduces delay tails and restores throughput under adversarial starvation, while staying nearly neutral under benign proposers, as shown in Appendix~\ref{app:supresults}, Figure~\ref{fig:wrapper_value_paired}.

\FloatBarrier

\begin{center}
\begin{minipage}{\linewidth}
  \centering
  \includegraphics[width=\linewidth]{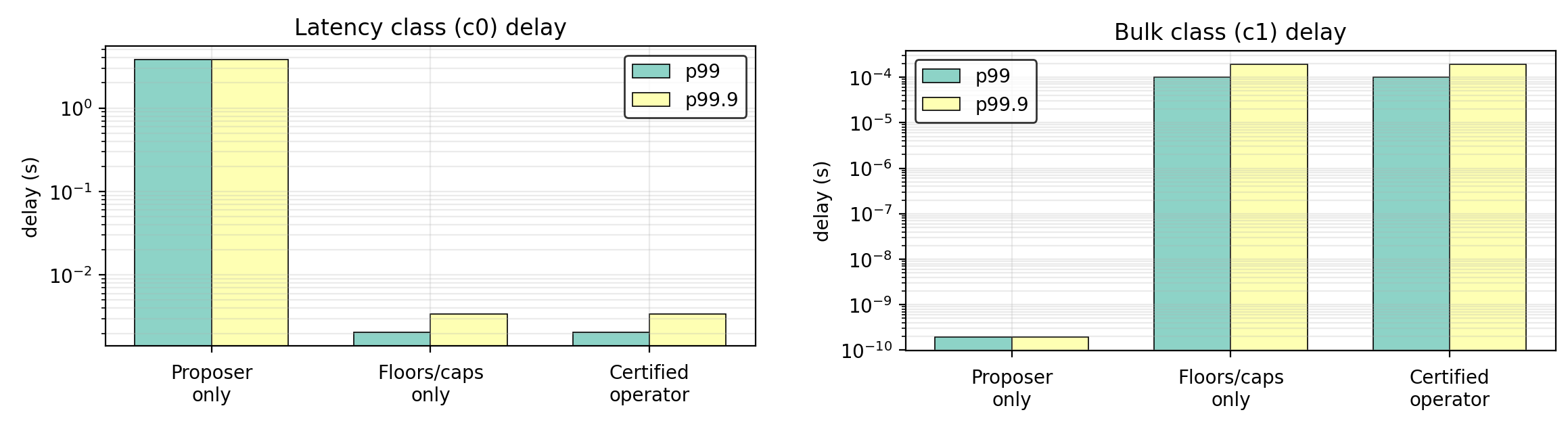}
  \captionof{figure}{%
  \textbf{Priority inversion under a poor proposer, and the effect of explicit floors and caps.}
  Tail delay for the latency class, c0, is shown at the top.
  Tail delay for the bulk class, c1, is shown at the bottom.
  Both are reported at p99 and p99.9 on a log scale.%
  }
  \label{fig:priority_inversion}
\end{minipage}
\end{center}

\FloatBarrier

\paragraph{Why floors and caps are first-class certificate constraints.}
Figure~\ref{fig:priority_inversion} shows a common deployment failure under a poor proposer: priority inversion.
When proposals are executed directly, the latency class c0 suffers much worse tail delay, while the bulk class c1 appears artificially healthy.
This is not a fine point of tuning.
It follows from leaving protection semantics buried inside proposer logic.
Explicit linear constraints, such as a per-tick service floor for protected classes \eqref{eq:floor} and optional caps for flagged classes \eqref{eq:cap_mitig}, prevent the inversion by reserving a minimum share of capacity for c0.
The certified operator enforces these constraints online, so the protection guarantee attaches to the executed action, not to the proposer’s good intentions.
Good intentions, in schedulers as elsewhere, are a poor substitute for a constraint.

\FloatBarrier

\begin{center}
\begin{minipage}{.85\linewidth}
  \centering
  \includegraphics[width=.85\linewidth]{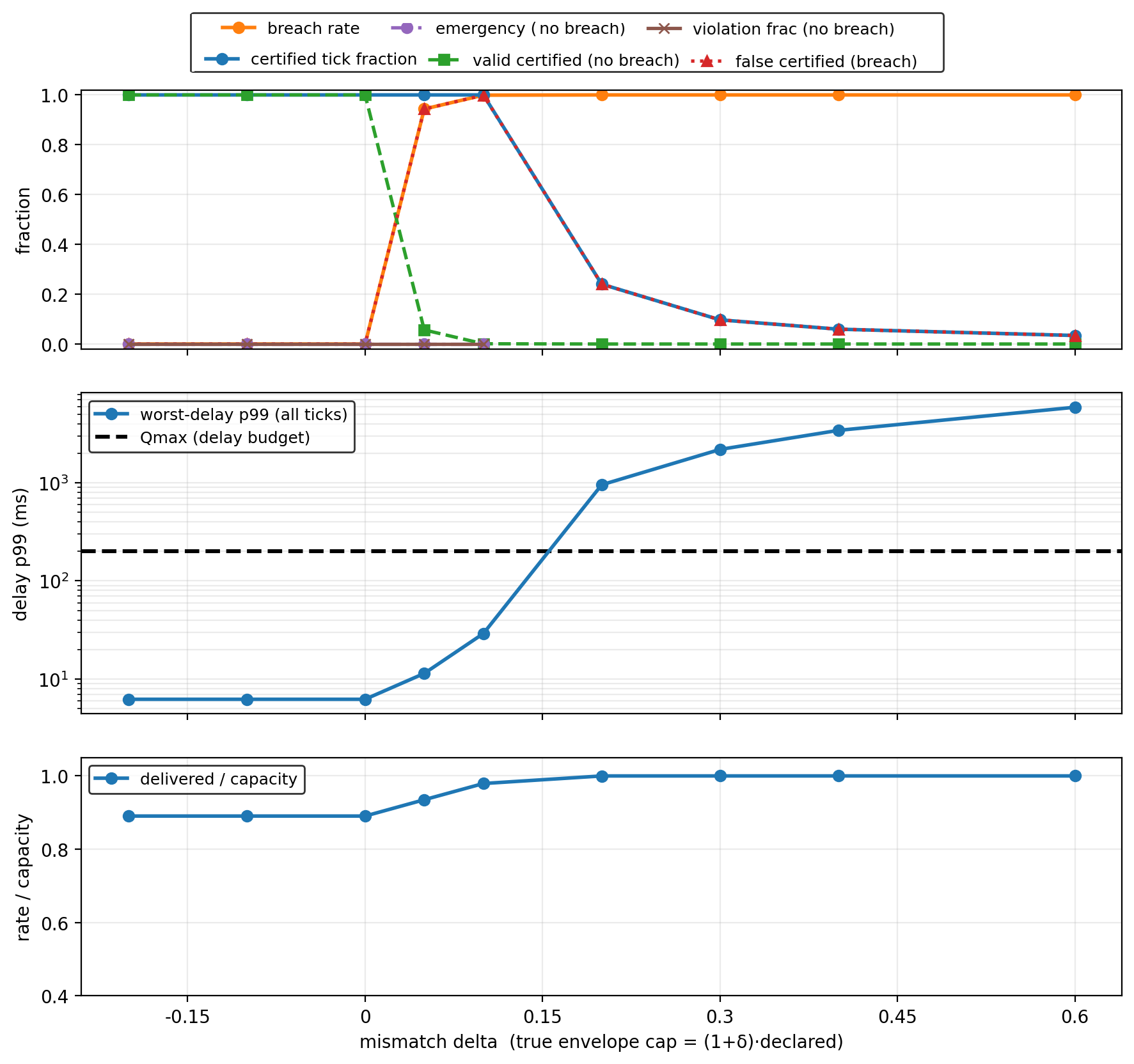}
  \captionof{figure}{%
  \textbf{Envelope mismatch sweep: breach visibility and status semantics.}
  Top: breach rate, defined as the fraction of ticks with \(b(t)=1\) in \eqref{eq:breach_indicator}, and certified-tick fraction.
  Middle: worst-queue delay p99 over all ticks, with delay budget \(Q^{\max}\).
  Bottom: delivered rate normalised by capacity.%
  }
  \label{fig:env_mismatch}
\end{minipage}
\end{center}

\FloatBarrier

\paragraph{Envelope mismatch is surfaced.}
Figure~\ref{fig:env_mismatch} sweeps the envelope mismatch parameter \(\delta\).
When \(\delta\le 0\), the declared envelope is pessimistic or correct. It becomes optimistic for \(\delta>0\).
For pessimistic or correct envelopes, breach is rare and most certified ticks are contract-valid. As the envelope becomes optimistic, the breach rate rises and the valid \texttt{CERTIFIED} fraction falls, because the contract assumption needed for safety and drift guarantees no longer holds. Some ticks may still satisfy the compiled constraints under the declared envelope, but they are not contract-valid and cannot support downstream composition. The middle panel shows the resulting delay growth.\FloatBarrier

\begin{center}
\begin{minipage}{.85\linewidth}
  \centering
  \includegraphics[width=.85\linewidth]{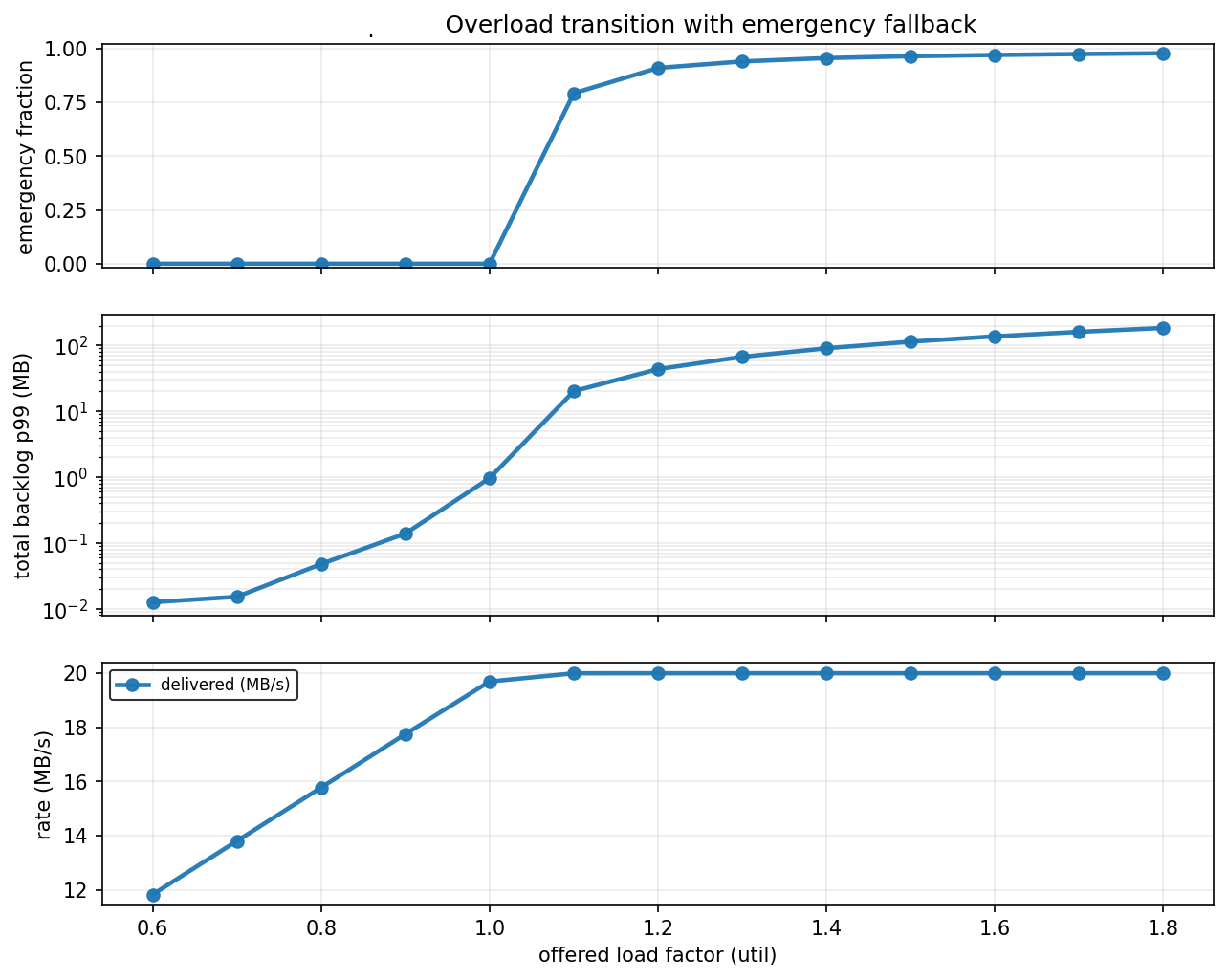}
  \captionof{figure}{%
  \textbf{Overload transition with emergency fallback.}
  Offered load is swept as a utilisation factor. The panels report emergency-mode fraction, total-backlog p99, and delivered rate.%
  }
  \label{fig:overload_emergency}
\end{minipage}
\end{center}

\FloatBarrier

\paragraph{Overload and emergency semantics.}
Infeasibility can occur even when envelopes are valid, for example when offered load exceeds capacity or configured caps are too tight. In this regime, the operator remains defined. It reports \texttt{INFEASIBLE} and returns a best-effort action with quantified slack, rather than silently violating the compiled constraints. Figure~\ref{fig:overload_emergency} shows the transition as offered load increases. Emergency-mode ticks rise sharply beyond the stability boundary, total-backlog p99 grows, and delivered rate saturates near capacity. This is the expected signature of overload, not a proposer artefact.
% =========================
\section{Conclusion}
% =========================

The central claim of this paper is that packet-network control should be certified at the point of execution. A proposer may be learned, heuristic, adversarial, or agentic, but the dataplane should act only on the certified operator's output. On \texttt{CERTIFIED} ticks, the guarantees attach to the executed action, the valid arrival envelope, the valid state bound, the service-tracking assumption, and the exported contract. When these conditions do not hold, the operator records the reason through breach, missing-information, infeasibility, and slack signals. An unsafe proposal should not become a safe action merely because the proposer was well designed.

The resulting framework brings backlog caps, service floors, mitigation caps, drift constraints, and envelope contracts into one per-tick certification interface. It also gives a compositional account: feed-forward networks compose through exported envelopes, while cyclic networks require a closure test under a small-gain condition. The evaluation supports this execution boundary under weak proposers, delayed telemetry, delayed actuation, envelope mismatch, overload, and millisecond-scale certification. These results should be read within that scope. They validate the certified operator in a byte-level closed-loop backend, while deployment-level scheduler tracking remains a separate systems question.

% The next step is to narrow that deployment gap. Scheduler tracking should be measured on real dataplanes, including Linux HTB, DRR, programmable switches, and hardware queues, so that \(\kappa_{\min}\) becomes a measured platform parameter rather than a modelling assumption. Finally, the trusted path from configuration to constraint compilation and solver output should be hardened through testing, translation validation, or verified compilation. These steps keep the paper's boundary intact: execute through a certified operator, state the assumptions under which the guarantee holds, and compose only over contracts that remain valid.

The framework also opens several natural directions for further work. One useful next direction is to instantiate the service-tracking factor on real dataplanes, including Linux HTB, DRR, programmable switches, and hardware queues. This would turn \(\kappa_{\min}\) from a conservative platform parameter into a measured quantity for specific scheduler regimes. Another direction is to strengthen the certification path itself, from configuration to constraint compilation and solver output, through systematic testing, translation validation, or verified compilation. 

% \subsection*{Author contributions: CRediT}
% Muhammad Bilal: Conceptualization, Methodology, Investigation, Visualization, Writing – original draft, Supervision. Jon Crowcroft: Investigation, Conceptualization, Validation, Writing – review \& editing. Xiaolong Xu: Investigation, Visualization, Validation, Writing – review \& editing. Huaming Wu: Conceptualization, Writing – review \& editing.  

\paragraph{\textbf{Ethical considerations}}
The study involved no human participants and used no sensitive or personally identifiable user data. No material ethical issues are therefore anticipated.

% \subsection*{Declaration of competing interest}
% The authors declare that they have no known competing financial interests or personal relationships that could have appeared to influence the work reported in this paper.

\subsection*{Declaration of generative AI and AI-assisted technologies in the manuscript preparation process}
During the preparation of this work the author(s) used AI-assisted tools for language editing and structural refinement. After using this tool/service, the author(s) reviewed and edited the content as needed and take(s) full responsibility for the content of the published article.

% =========================
% Bibliography
% =========================
\bibliographystyle{plain}
\bibliography{reference}

\appendix
\section{Appendices}
\subsection{Algorithms }
\label{app:algorithms}
% -------------------------

This appendix gives the auxiliary procedures used by the main text. Algorithm~\ref{alg:dag} gives the executable form of DAG envelope propagation used in Theorem~\ref{thm:dag_comp}: upstream modules export bounds, and downstream modules aggregate them in topological order. Algorithm~\ref{alg:closure} gives the cyclic counterpart, where the inflow envelope is found by fixed-point iteration and is valid only when the small-gain test succeeds. Algorithm~\ref{alg:stress} summarises the stress harness used in Section~\ref{sec:experiments}, including delayed telemetry, delayed actuation, timing-shaped arrivals, breach indicators, and status traces.

\begin{algorithm}[h]
\caption{Envelope propagation in a DAG, one forward pass}
\label{alg:dag}
\small
\raggedright
\begin{algorithmic}[1]
\Require DAG modules in topological order \(M_1\prec M_2\prec\cdots\prec M_L\), exogenous envelopes \(\bar a^{M_\ell}(t)\)
\State Set \(\bar{\bm{\mathrm{in}}}^{M_\ell}(t)\gets 0\) for all \(\ell\)
\For{\(\ell=1\) to \(L\)}
  \State Run the module step, Algorithm~\ref{alg:pcp}, to obtain exported envelopes \(\bar z_{M_\ell\to V}(t)\)
  \ForAll{outgoing neighbours \(V\) of \(M_\ell\)}
    \State \(\bar{\bm{\mathrm{in}}}^V(t)\gets \bar{\bm{\mathrm{in}}}^V(t)+\bar z_{M_\ell\to V}(t)\)
  \EndFor
\EndFor
\end{algorithmic}
\end{algorithm}

\begin{algorithm}[t]
\caption{Cyclic envelope closure by fixed-point iteration}
\label{alg:closure}
\small
\raggedright
\begin{algorithmic}[1]
\Require Exogenous envelope \(\bar{\bm a}(t)\), maps \(F\) and \(G\), tolerance \(\eta>0\), maximum iterations \(K_{\max}\)
\State Initialise \(\bar{\bm{\mathrm{in}}}^{(0)}(t)\gets 0\)
\For{\(k=0\) to \(K_{\max}-1\)}
  \State \(\bar{\bm z}^{(k)}(t)\gets F(\bar{\bm a}(t)+\bar{\bm{\mathrm{in}}}^{(k)}(t))\)
  \State \(\bar{\bm{\mathrm{in}}}^{(k+1)}(t)\gets G(\bar{\bm z}^{(k)}(t))\)
  \If{\(\|\bar{\bm{\mathrm{in}}}^{(k+1)}(t)-\bar{\bm{\mathrm{in}}}^{(k)}(t)\|\le \eta\)}
    \State \textbf{break}
  \EndIf
\EndFor
\State \textbf{Output}: \(\bar{\bm{\mathrm{in}}}(t)\gets \bar{\bm{\mathrm{in}}}^{(k+1)}(t)\) (closed envelopes)
\end{algorithmic}
\end{algorithm}

\begin{algorithm}[t]
\caption{Stress harness for delayed telemetry, delayed actuation, and timing-shaped arrivals}
\label{alg:stress}
\small
\raggedright
\begin{algorithmic}[1]
\Require Envelope sequence \(\bar A(t)\), telemetry delay \(\tau_y\), actuation delay \(\tau_u\), capacity \(C(t)\), horizon \(T\)
\State Initialise queues \(q(0)\) and an action buffer of length \(\tau_u\)
\For{\(t=0\) to \(T-1\)}
  \State arrivals \(A(t)\) within or beyond \(\bar A(t)\)
  \State Form telemetry \(y(t)=Hq(t-\tau_y)+\nu(t)\)
  \State Proposer outputs \(\tilde u(t)\) from \(y(\le t)\)
  \State Operator sets \((u(t),\sigma(t))\leftarrow\mathcal{C}_{\Theta}(\tilde u(t),y(t))\), while baselines set \(u(t)\leftarrow\tilde u(t)\)
  \State Apply delayed action \(u_{\mathrm{applied}}(t)=u(t-\tau_u)\)
  \State Update queues and log violations, breach flags, status flags, and runtime
\EndFor
\end{algorithmic}
\end{algorithm}

\subsection{Supportive Results}
\label{app:supresults}
% -------------------------

\FloatBarrier

\begin{center}
\begin{minipage}{\linewidth}
  \centering
  \includegraphics[width=\linewidth]{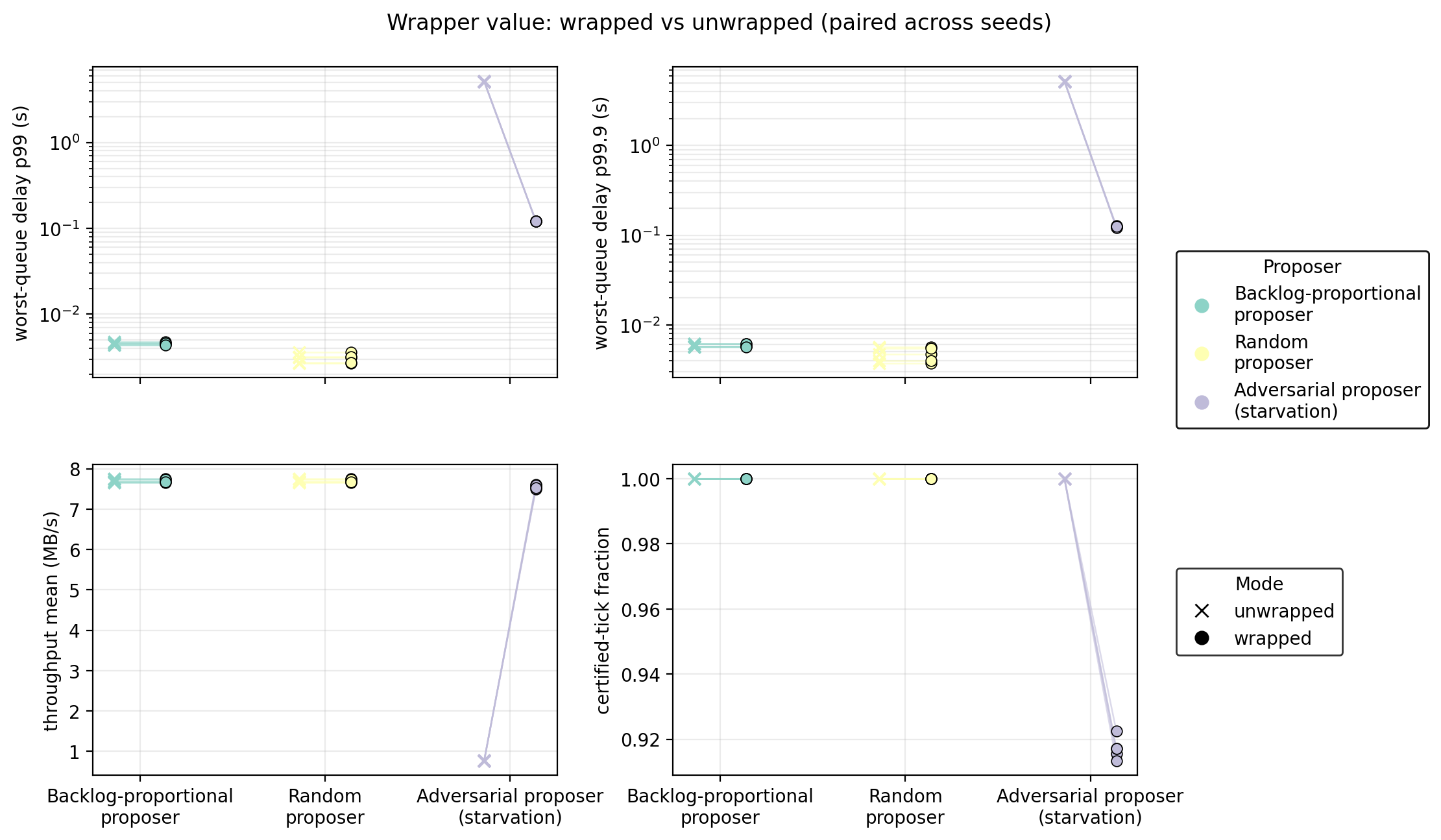}
  \captionof{figure}{%
  \textbf{Certified-operator value per seed, paired with and without certification.}
  Each proposer is evaluated on the same seeds in wrapped and unwrapped modes, with lines connecting paired runs.
  The paired view shows that improvements under the adversarial proposer are consistent across seeds for both delay tails and throughput.
  It also shows that the wrapped mode’s \texttt{CERTIFIED} fraction remains below one in this regime.%
  }
  \label{fig:wrapper_value_paired}
\end{minipage}
\end{center}

\FloatBarrier

\paragraph{Seed-paired evidence for wrapper value.}
Figure~\ref{fig:wrapper_value_paired} provides a seed-paired robustness check. Certified and direct-execution runs use identical random seeds. For backlog-proportional and random proposers, the paired points nearly coincide, indicating that certification is mostly non-invasive when proposals are already well behaved. Under the adversarial starvation proposer, certification reduces delay tails and restores throughput. Non-certified ticks in these runs are mainly \texttt{INFEASIBLE} ticks caused by tight compiled lower bounds under stress. Envelope breaches are analysed separately in the envelope-mismatch experiment.

\end{document}